\documentclass[12pt, a4paper]{article}
\usepackage{jheparxiv}
\usepackage[utf8]{inputenc}
\usepackage{amsmath}
\usepackage{amsfonts}
\usepackage{amssymb}
\usepackage{bbold}
\usepackage{latexsym}
\usepackage{mathrsfs}
\usepackage{braket}		
\usepackage{wasysym}
\usepackage{oplotsymbl}
\usepackage{graphicx}

\usepackage{empheq}

\newcommand{\Mat}{\mathtt{Mat}}
\newcommand{\mso}{\mathfrak{so}}

\newcommand{\nn}{\nonumber}

\newcommand{\cHbold}{\boldsymbol{\mathcal{H}}}

\usepackage{tikz-cd}
\usepackage{graphicx}
\usepackage{color}
\usepackage{xcolor}
\usepackage{slashed}
\usepackage{twistor}
\usepackage[all]{xy}
\usepackage{tikz-cd}
\usepackage{mathtools}
\usepackage{float}
\usetikzlibrary{arrows, matrix}
\tikzcdset{arrow style=math font}
\tikzset{>=latex}

\newcommand{\sQ}{\mathsf{Q}}
\newcommand{\sP}{\mathsf{P}}
\newcommand{\sW}{\mathsf{A}}

\newcommand{\uomega}{\underline{\omega}}
\newcommand{\usC}{\underline{\mathsf{C}}}

\newcommand{\tp}{\mathtt{p}}

\newcommand{\tH}{\mathtt{H}}
\newcommand{\K}{\mathbb{K}}

\usepackage{skak}

\newcommand{\sL}{\mathsf{L}}

\newcommand{\Omegabold}{\boldsymbol\Omega}

\newcommand{\gbold}{\boldsymbol{g}}

\newcommand{\msp}{\mathfrak{sp}}
\newcommand{\eps}{\epsilon}

\newcommand{\msf}[1]{\mathsf{#1}}

\newcommand{\sF}{\mathsf{F}}

\newcommand{\sC}{\mathsf{C}}

\newcommand{\varpibold}{\boldsymbol{\varpi}}

\newcommand{\ebold}{\boldsymbol{e}}

\newcommand{\Hbold}{\boldsymbol{H}}
\newcommand{\Rbold}{\boldsymbol{R}}

\title{\Large Self-dual pp-wave solutions in chiral higher-spin gravity}

\author[1]{Tung Tran}

\affiliation[1]{Asia Pacific Center for Theoretical Physics, POSTECH, 77 Cheongamro, Nam-gu,
Pohang-si, Gyeongsangbuk-do, 37673, Korea}

\emailAdd{tran.tung@apctp.org}

\abstract{We show that chiral higher-spin gravity with a vanishing cosmological constant  admits a class of exact self-dual pp-wave solutions derived from harmonic scalar functions and two principal spinors. These solutions satisfy both the linear and non-linear equations of motion, as they annihilate all higher-order vertices, leading to the equations of motion for free fields on a self-dual background sourced by a positive-helicity spin-2 field. Our method employs a simple Kerr-Schild ansatz for positive-helicity chiral higher-spin fields 
adapted to the self-dual gravity framework.}

\begin{document}

\maketitle
  

\section{Introduction}

By extending the usual GR with massless higher-spin fields, one can enlarge the symmetries of gravity by those associated with higher-spin fields, see e.g. \cite{Bekaert:2022poo} for an overview. This led to various higher-spin gravities (HSGRA)s, which serve as simple toy models toward a theory of quantum gravity, one that is presumably UV-finite and free of black hole-like singularities. 

Intriguingly, to construct theories of interacting massless higher-spin fields, one often needs to relax some of the usual assumptions underlying field theories such as unitarity or 
parity invariance, which, to some extent, is not a surprise fact (see discussion in \cite{Weinberg:1964ew,Coleman:1967ad,Maldacena:2011jn,Bekaert:2015tva,Sleight:2017pcz,Tran:2022amg,Neiman:2023orj}).

At the moment, QFT-compatible HSGRAs with propagating degrees of freedom are either higher-spin extension of (self-dual) Weyl gravity \cite{Segal:2002gd,Tseytlin:2002gz,Bekaert:2010ky,Haehnel:2016mlb,Adamo:2016ple,Basile:2022nou}, or chiral higher-spin theories, see e.g. \cite{Metsaev:1991mt,Metsaev:1991nb,Ponomarev:2016lrm,Metsaev:2018xip,Skvortsov:2018uru,Metsaev:2019dqt,Metsaev:2019aig,Tsulaia:2022csz}. There are also in $3d$ HSGRA \cite{Blencowe:1988gj, Bergshoeff:1989ns, Campoleoni:2010zq, Henneaux:2010xg, Pope:1989vj, Fradkin:1989xt, Grigoriev:2019xmp, Grigoriev:2020lzu,Sharapov:2024euk} and a 
higher-spin extension of Jackiw-Teitelboim gravity in $2d$ \cite{Alkalaev:2019xuv}. Another interesting example is the higher-spin gauge theory induced by the IKKT matrix model on an FLRW cosmological background with a truncated higher-spin spectrum \cite{Sperling:2019xar,Steinacker:2021yxt}.\footnote{See e.g. \cite{Steinacker:2023cuf,Steinacker:2024huv,Manta:2024vol} and references therein for the recent development of this model.}

In this work, we focus on chiral higher-spin gravity and explore the space of its exact solutions to examine how much they differ from GR.\footnote{See \cite{Skvortsov:2024rng} for the lower-spin BPST solution in chiral HSGRA.} Note that one of the key features of chiral HSGRA is that its free action can be defined on any self-dual background 
\cite{Krasnov:2021nsq}, i.e. background with Euclidean or split signature where half of the component of the Weyl tensor (the anti self-dual part) is zero. However, whether this holds with all interactions present remains unclear, making verification essential. 

Henceforth, we work with a complexified setting, which provides more flexibility in doing computation. This philosophy, which has deep roots in twistor theory \cite{Penrose:1967wn}, is what we want to employ in this paper in finding exact solutions of chiral HSGRA using its equations of motion \cite{Sharapov:2022faa,Sharapov:2022awp}.

The rest of the paper is organized as follows:
\begin{itemize}
    \item[-] In Section \ref{sec:review}, we present crucial information, including the spectrum and explicit forms of local vertices, to study chiral higher-spin gravity (HSGRA) through its equations of motion derived from certain minimal model with local $L_{\infty}$-structures of a certain $Q$-manifold \cite{sullivan1977infinitesimal,Alexandrov:1995kv}. 
    \item[-] Section \ref{sec:chiral-FDA} studies the exact solutions of chiral HSGRA sourced by a set of harmonic scalar functions $\{\tH_s\}$ and the principal spinors $(\lambda,\tilde\lambda)$ \`a la Newman-Penrose. Since all local vertices of chiral HSGRA are smooth in the cosmological constant $\Lambda$, as shown in \cite{Sharapov:2022awp}, 
    it seems reasonable to first explore the case $\Lambda=0$. 
    This is exactly what we do in this paper. 
    
    The spin-2 sector of the theory needs some special attention, since it generates the proposed self-dual pp-wave background away from a flat background. In particular, we will adapt the Cartan structure equations to the context of self-dual GR as discussed in \cite{Krasnov:2021cva}, which admits a smooth deformation from flat space to (A)dS. By considering an appropriate harmonic ansatz for fluctuation, akin to those in \cite{Adamo:2022lah,Neiman:2023bkq}, we show that the spin-2 sector of chiral HSGRA indeed admit a class of exact solutions, which solve the full non-linear EoMs, for any harmonic profile function $\tH_2$. From here, the construction for exact higher-spin solutions is completely analogous. Note that our solution has a simple form, and it makes most of the structure maps of the underlying $L_{\infty}$-algebra 
    vanish, see earlier discussion in \cite{Sharapov:2022faa}. Moreover, it is reasonable to speculate that once the space of solutions is found, it should be `portable' to other nearby local higher-spin theories \cite{Ponomarev:2017nrr,Monteiro:2022xwq} as well. 

    In the concluding part of Section \ref{sec:chiral-FDA}, we also derive the effective action corresponding to our exact solutions and show that this action indeed reduces to the standard kinetic term for free spinning fields propagating on the proposed self-dual pp-wave background. That is, despite having non-trivial interactions, the theory appears to be free when the scattering data are those obtained exact solutions.

    \item[-] We end with a discussion in Section \ref{sec:discussion}. There are also appendices which the reader can refer to in due time.
\end{itemize}

\paragraph{Notation.} We use lowercase Roman letters $a,b=1,2,3,4$ to denote tangent indices, while uppercase letters $A,B=1,2,3,4$ will be used for twistor or $\msp(4)$-indices. The Greek letters $\alpha,\beta=0,1$ and $\dot\alpha,\dot\beta=\dot 0,\dot 1$ are used to for spinorial indices. Symmetrized indices are denoted by the same letters, e.g. $A_{a}B_{a}=\frac{1}{2}(A_{a_1}B_{a_2}+A_{a_2}B_{a_1})$ while fully symmetric rank-s tensor will be denoted by $T_{a(s)}=T_{a_1\ldots a_s}$.

\section{Review}\label{sec:review}


This section provides the mathematical background used to construct explicit representations of chiral HSGRA with vanishing cosmological constant. We shall be brief and refer the reader to \cite{Sharapov:2022faa,Sharapov:2022awp} for the details and to \cite{Sharapov:2022wpz,Sharapov:2022nps} for a recent development.

\subsection{Free differential algebra}\label{sec:FDA}

Let $\cM$ be a $4d$ spacetime manifold and $C^{\infty}(T[1]\cM)\equiv \Omega^{\bullet}(\cM)$ be its associated graded commutative algebra of differential forms endowed with the exterior differential $d$, which squares to zero. Then, the pair $(T[1]\cM,Q_{\cM}\equiv d)$ present the simplest example of a $Q$-manifold.

Now, consider differential-preserving maps $\Phi^A(x,dx):T[1]\cM\rightarrow \cN$ from the source $T[1]\cM$ to another target $Q$-manifold $(\cN,Q_{\cN})$, which play the role of fields on $\cN$. 
Then, around a stationary point $\tp\in \cN$, where $Q_{\cN}(\tp)=0$, the associated set of pde. for $\Phi^A$ takes the form
\begin{align}\label{dPhi=QPhi}
    d \Phi^A &= Q_{\cN}^A(\Phi)\,.
\end{align}
This is also known as free differential algebra (FDA) \cite{sullivan1977infinitesimal,van2005free,DAuria:1980cmy,Vasiliev:1988sa}. The nilpotency of $Q_{\cN}\equiv Q$ is equivalent to the statement:
\begin{align}\label{QQ=0}
    Q^B\frac{\p}{\p \Phi^B}Q^A=0\qquad \Leftrightarrow \qquad \sum \pm Q^{B}{}_{M_1\ldots M_n}Q^{A}{}_{N_1\ldots B\ldots N_k}=0\,.
\end{align}
This gives us an $L_{\infty}$-algebra \cite{Alexandrov:1995kv,stasheff2006homotopy} whose $L_{\infty}$-structure maps are the Taylor's coefficients $Q^B_{M_1\ldots N_n}$.

In the simplest setting, which works for generic gauge theories and gravity, we can choose $\Phi(x)=(\omega(x),\sC(x))$ to be the graded coordinates of $\cN$. Here, $\omega$ is a gauge field of degree one taking values in some Lie algebra $\mathfrak{g}$, and $\sC$ is a `matter' field of degree zero taking values in some representation $\varrho(\mathfrak{g})$ of $\mathfrak{g}$. Then, the equations of motion resulting from \eqref{dPhi=QPhi} with the above field content, in general, are 
\begin{subequations}\label{eq:FDA}
\begin{align} d\omega&=\cV(\omega,\omega)+\cV(\omega,\omega,\sC)+\cV(\omega,\omega,\sC,\sC)+\ldots\,,\\
    d \sC&=\cU(\omega,\sC)+\cU(\omega,\sC,\sC)+\cU(\omega,\sC,\sC,\sC)+\ldots\,.
\end{align}
\end{subequations}
where the notations $(\cV,\cU)$ are obvious. To sum things up, we can define a theory whenever we can specify $\{(\omega,\sC)\}$ and $\{\cV_n(\omega,\omega,\sC,\ldots)\,,\cU_n(\omega,\sC,\ldots)\}$.\footnote{The subscripts of $\cV$ and $\cU$ denote the number of arguments that enter the corresponding vertices.}

Note that \eqref{eq:FDA} can be referred to as a deformation of the following differential graded Lie algebra (dgLa)
\begin{align}\label{eq:linear-FDA}
    d\omega=\frac{1}{2}[\omega,\omega]\,,\qquad d \sC=\rho(\omega)\sC\,,
\end{align}
where now $\rho(\omega)$ denotes some appropriate representation of $\omega$. 
The first equation is simply the definition of a flat connection associated to the one-form $\omega$, while the second expresses the covariant constancy of the 0-form $\sC$ wrt. the flat connection.

\subsection{Chiral HSGRA}
\subsubsection{Chiral FDA data}
The local isomorphism $SO(4,\C)\simeq SL(2,\C)\times SL(2,\C)$ in $4d$ complexified spacetime allows us to write any tensorial fields in terms of spin-tensor fields $T_{\alpha(m)\,\dot\alpha(n)}$.

In the case of chiral HSGRA, $(\omega,\sC)$ can be written in terms of generating functions as
\small
\begin{subequations}\label{eq:generatingfunctions}
\begin{align}
    \text{1-form}&:  &\omega&=\sum_s\frac{1}{(2s)!}\omega_{A(2s)}Y^{A(2s)}=\sum_{m,n\geq 0}\frac{1}{m!n!}\omega_{\alpha(m)\,\dot\alpha(n)}\, y^{\alpha(m)}\tilde y^{\dot\alpha(n)}\,, \qquad m+n\in 2\N\,,\\
    \text{0-form}&:  &\sC&=\sum_s\frac{1}{(2s)!}\sC_{A(2s)}Y^{A(2s)}=\sum_{m,n\geq 0}\frac{1}{m!n!}\sC_{\alpha(m)\,\dot\alpha(n)}\,y^{\alpha(m)}\tilde y^{\dot\alpha(n)}\,, \qquad m+n\in 2\N\,.
\end{align}
\end{subequations}
\normalsize
where $Y^A=(y^{\alpha},\tilde y^{\dot\alpha})$. Here, $y^{\alpha},\tilde y^{\dot\alpha}$ are commutative spinorial variables, and $y^{\alpha(n)}=y^{\alpha_1}\ldots y^{\alpha_n}$ etc. is our standard abbreviation. The dynamical bosonic fields, which give us free action, are singled out as
\begin{subequations}\label{eq:generatingAB}
    \begin{align}
   \text{helicity $+s$}&: \qquad &\cA:&=\sum_{s\in \N}\frac{1}{(2s)!}\cA_{\alpha(2s)}y^{\alpha(2s)}=\omega(y,\tilde y=0)\,,\\ 
   \text{helicity $-s$}&:\qquad &\cB:&=\sum_{s\in \N}\frac{1}{(2s)!}B_{\alpha(2s)}y^{\alpha(2s)}=\sC(y,\tilde y=0)\,,\\
   \text{helicity  \ \, $0$}&:\qquad &\varphi:&=\sC(y=0,\tilde y=0)\,.
\end{align}
\end{subequations}
Using $B^{\alpha(2s)}$ and $\cA_{\alpha(2s-2)}$ one can write down the free action for chiral spinning fields as \cite{Krasnov:2021nsq}
\begin{align}\label{eq:free action}
    S=\int B^{\alpha(2s)}\wedge H_{\alpha\alpha}\nabla \cA_{\alpha(2s-2)}\,,\qquad H_{\alpha\alpha}:=h_{\alpha\dot\gamma}\wedge h_{\alpha}{}^{\dot\gamma}\,.
\end{align}
In the above, $h_{\alpha\dot\alpha}$ is the background vierbein of any SD spacetime where the ASD component $C_{\alpha(4)}$ of the Weyl tensor vanish, and $\nabla$ is the associated self-dual connection. This action is invariant on any self-dual background (with or without the cosmological constant) under $\delta \cA_{\alpha(2s-2)}=\nabla \xi_{\alpha(2s-2)}+h_{\alpha}{}^{\dot\beta}\eta_{\alpha(2s-3)\,\dot\beta}$, which can be checked immediately by noting that $\nabla^2\xi_{\alpha(2s-2)}=H_{\alpha}{}^{\beta}\xi_{\alpha(2s-3)\beta}$ and $H_{\alpha\alpha}\wedge h_{\alpha}{}^{\dot\beta}=0$. 

As explained in \cite{Krasnov:2021nsq}, the physical one-form $\cA_{\alpha(2s-2)}\equiv \omega_{\alpha(2s-2)}$ can be written as
\begin{align}
    \omega_{\alpha(2s-2)}=h^{\beta\dot\beta}A_{\alpha(2s-2)\beta\,\dot\beta}+h_{\alpha}{}^{\dot\alpha}\vartheta_{\alpha(2s-3)\,\dot\alpha}\,.
\end{align}
Since we can use the gauge transformation $\delta_{\eta} \omega_{\alpha(2s-2)}=h_{\alpha}{}^{\dot\alpha}\eta_{\alpha(2s-3)\,\dot\alpha}$ to gauge away the second component $\vartheta$ in $\omega^{\alpha(2s-2)}$, we will simply consider for most of the time 
\begin{align}\label{eq:simplify-omega}
    \omega_{\alpha(2s-2)}=h^{\beta\dot\beta}A_{\alpha(2s-2)\beta\,\dot\beta}.
\end{align}
Now, at the free level the FDA for chiral HSGRA for $\Lambda=0$ is simple
\begin{subequations}\label{linear-FDA-general}
    \begin{align}
        \nabla\omega_{\alpha(2s-2-k)\,\dot\alpha(k)}&=h_{\alpha}{}^{\dot\beta}\omega_{\alpha(2s-3-k)\,\dot\alpha(k)\dot\beta}\,,\qquad \qquad \quad 0\leq k\leq 2s-2\,,\\
        \nabla\omega_{\dot\alpha(2s-2)}&=\widetilde H^{\dot\beta\dot\beta}C_{\dot\alpha(2s-2)\dot\beta\dot\beta}\,,\qquad \qquad \qquad \widetilde H^{\dot\alpha\dot\alpha}:=h_{\gamma}{}^{\dot\alpha}\wedge h^{\gamma\dot\alpha}\,,\\
        \nabla C_{\alpha(k)\,\dot\alpha(2s+k)}&=h^{\beta\dot\beta}C_{\alpha(k)\beta\,\dot\alpha(2s+k)\dot\beta}\,,\\
        \nabla B^{\alpha(2s+k)\,\dot\alpha(k)}&=h_{\beta\dot\beta}B^{\alpha(2s+k)\beta\,\dot\alpha(k)\dot\beta}\,,
    \end{align}
\end{subequations}
where $C_{\alpha(k)\,\dot\alpha(2s+k)}$ are the $\nabla$-generated auxiliary fields of $C_{\dot\alpha(2s)}$, and $B^{\alpha(2s+k)\,\dot\alpha(k)}$ are the $\nabla$-generated auxiliary fields of $B^{\alpha(2s)}$, cf. \eqref{eq:generatingAB}.

Since the next (auxiliary) fields can be generated by the previous (auxiliary) fields by acting $\nabla$ on the physical fields again and again, the whole spectrum of fields can be obtained starting from the \emph{chiral fields} $\cA_{\alpha(2s-2)}$ and $B^{\alpha(2s)}$ and the scalar field $\varphi$ (see Fig. \ref{chiral-spectrum})\,. 
\begin{figure}[ht!]
    \centering
    \includegraphics[scale=0.42]{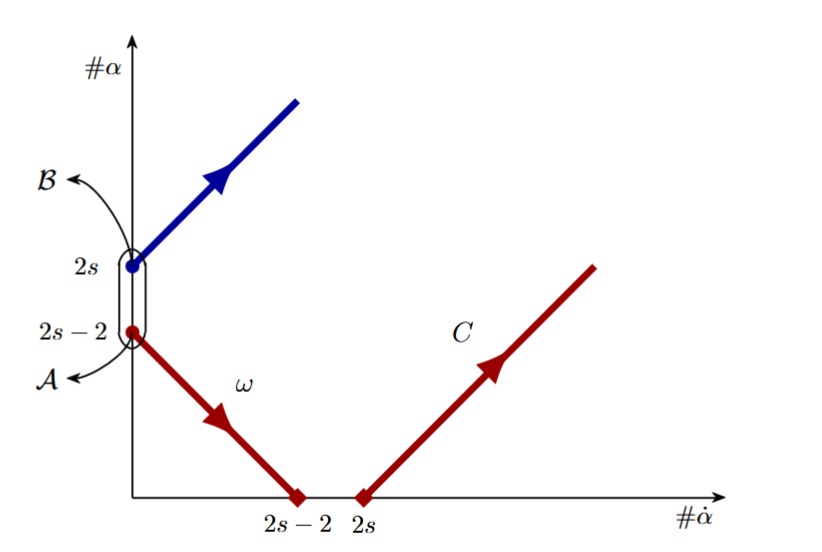}
    \caption{The horizontal/vertices axes represent the number of dotted and un-dotted spinorial indices that a tensorial field has. Here, [red] are fields with positive helicity, and [blue] are fields with negative helicity. The arrows indicate the directions in which auxiliary fields will be generated by acting $\nabla$ on the previous ones. All fields generated this way are referred to as chiral FDA data.}
    \label{chiral-spectrum}
\end{figure}

\subsubsection{Vertices}
The next ingredients we need are the structure maps/vertices $(\cV,\cU)$ in \eqref{eq:FDA}. Conveniently, we can write all vertices of chiral HSGRA as poly-differential operators:
\begin{align}\label{polydiff}
   \sF(\Phi_1,\ldots,\Phi_n):=\sF(y,\tilde y,\p_{i\,\alpha},\tilde\p_{i\,\dot\alpha})\Phi_1(y_1,\tilde y_1)\ldots \Phi_n(y_n,\tilde y_n)\Big|_{y_i,\tilde y_i=0}\,,\quad \Phi=(\omega,\sC)\,.
\end{align}
For conciseness, we will suppress the notation $\big|_{y_{i}=0,\,\tilde y_i=0}$ at the end of each expression. Here, $\sC=(C,B)$ is the master 0-form field, which contains both negative-helicity $B$'s and positive-helicity $C$'s fields and the scalar field $\varphi$. 

Below, we present some of the key features of these structure maps. The technical details in deriving these maps can be found in \cite{Sharapov:2022faa,Sharapov:2022awp} (see also Appendix \ref{app:A} for a recap).


\paragraph{$\star$-product} The first structure map in the minimal model of interest is the binary structure map $\sF_2$. It is given by a graded associative star-product:
\begin{align}\label{eq:star-product}
    \Phi_1\star \Phi_2&= \exp\Big(Y^A\frac{\p}{\p Y_1^A}+Y^A\frac{\p}{\p Y_2^A}-I^{AB}\frac{\p}{\p Y_1^A}\frac{\p}{\p Y_2^B}\Big)\Phi_1(Y_1) \Phi_2(Y_2)\nn\\
    &=\exp\Big(\langle y\,\p_{1}\rangle+\langle y\,\p_{2}\rangle+\Lambda \langle\p_1\,\p_{2}\rangle+[ y\,\p_{1}]+[y\,\p_{2}]+[\p_1\,\p_2]\Big)\Phi_1\Phi_2\,,
\end{align}
where $\Lambda$ is the cosmological constant, $I^{AB}$ is the infinity twistor \cite{Penrose:1967wn} obeying
\begin{align}
    I^{AB}=\begin{pmatrix} \Lambda \epsilon^{\alpha\beta} & 0 \\
    0 & \epsilon^{\dot\alpha\dot\beta}\end{pmatrix}\,,\qquad I_{AB}=\begin{pmatrix} \epsilon_{\alpha\beta} & 0 \\
    0 & \Lambda \epsilon_{\dot\alpha\dot\beta}\end{pmatrix}\,,\qquad I_{AC}I^{BC}=\Lambda \delta_A{}^B\,.
\end{align}
Note that 
\begin{align}
    \langle y \,\p_i\rangle = y^{\alpha}\p_{i\alpha}\,,\quad \langle \p_i\,\p_j\rangle=\p_i^{\alpha}\p_{j\alpha}\,; \qquad  [y\, \p_i]=\tilde y^{\dot\alpha}\tilde\p_{i\dot\alpha}\,,\quad [\p_i\,\p_j]=\tilde\p_{i}^{\dot \alpha}\tilde\p_{j\dot \alpha}\,,
\end{align}
where $\p_{i\alpha}=\frac{\p}{\p y_i^{\alpha}}$, $\tilde\p_{i\dot\alpha}=\frac{\p}{\p \tilde y_i^{\dot\alpha}}$. Here, $e^{\langle y\,\p_i\rangle}$ may be referred to as a shift operator with the property $e^{\langle y\,\p_i\rangle}\Phi(y_i)=\Phi(y_i+y)$. We will use the following convention to raise and lower spinorial indices
\begin{align}
    v^{\alpha}=v_{\beta}\epsilon^{\alpha\beta}\,,\quad v_{\alpha}=v^{\beta}\epsilon_{\beta\alpha}\,,\quad v^{\dot\alpha}=v_{\dot\beta}\epsilon^{\dot\alpha\dot\beta}\,,\quad v_{\dot\alpha}=v^{\dot\beta}\epsilon_{\dot\beta\dot\alpha}\,.
\end{align}
The Weyl algebra $\sW_2$ of polynomial functions in $Y^A=(y^{\alpha},\tilde y^{\dot\alpha})$ defined in terms of the $\star$-product cf. \eqref{eq:star-product} admits the decomposition $\sW_2=\sW_1(y)\otimes \widetilde\sW_1(\tilde y)$ where the subscript denotes the number of canonical pairs. 
We then define
\begin{align}
    \hs_{\Lambda}:=\sW_1\otimes \big(\widetilde \sW_1\otimes \Mat(n,\K)\big)\,,
\end{align}
to be the underlying higher-spin algebra of the system. Note that when $\Lambda=0$, the Poisson structure defined by the infinity twistor is degenerate, and $\sW_1$ becomes a commutative algebra while $\widetilde\sW_1\otimes \Mat(n,\K)$ remains to be an non-commutative (NC) algebra. This is the setting that we consider in this work.





\paragraph{Higher-order vertices.} The details of deriving higher-order vertices for chiral FDA, cf. \eqref{eq:FDA}, can be found in e.g. \cite{Sharapov:2022nps}. Here, we shall emphasize some of the main points, which are relevant for our analysis in Section \ref{sec:chiral-FDA}. Henceforth, we focus on the case where $\Lambda=0$.  

$\diamond$ \underline{$\cV_{n\geq 3}(\omega,\omega,\sC,\ldots,\sC)$ vertices.} The $\cV_n(\omega,\omega,\sC,\ldots)$ maps consist of $\binom{n}{2}=n(n-1)/2$ sub-vertices:
\begin{align}
    \cV_n(\omega,\omega,\sC,\ldots,\sC)=V^{(n)}_{\omega,\omega,\sC,\ldots,\sC}+V^{(n)}_{\omega,\sC,\omega,\sC,\ldots,\sC}+\ldots
\end{align}
We shall denote $i$ and $j$ to be the positions of $\omega$'s in the sub-vertices $V$'s (assuming $i<j$) since they can be useful references for describing the sub-vertices $V$'s. It is worth noting that $\cV_n$ are $L_{\infty}$ structure maps obtained by suitable graded symmetrization of a certain $A_{\infty}$ structure maps/sub-vertices $V^{(n)}$ (see e.g. Appendix \ref{app:A}).

For instance, a vertex of the form $V^{(n)}_{\sC,\ldots,\omega_i,\sC,\ldots,\omega_j,\ldots,\sC}$ has the following expressions \cite{Sharapov:2022wpz}:
\begin{align}\label{V-subvertices}
        V^{(n)}_{\sC,\ldots,\omega_i,\sC,\ldots,\omega_j,\ldots,\sC}=(-)^{i+j-1}\sigma_{G_{ij}}\,\mho_n \langle \p_i\,\p_j\rangle^{n-2}\int_{\Delta_{n-2}}\exp\tr\,\Big(\sP_{G_{ij}}\sQ^T
        \Big)\,,
    \end{align}
    where 
    \small
    \begin{subequations}
        \begin{align}
        \sQ&=\begin{pmatrix}
            u_1 & u_2 & \ldots & u_{n-2} & 1-\sum u_i \\
           v_1 & v_2 & \ldots & v_{n-2} & 1-\sum v_i
        \end{pmatrix}\,,\\
        \sP_{G_{ij}}&=\begin{pmatrix}
          \langle \p_i\,\p_1 \rangle & \ldots & -\langle \p_i \,\p_j\rangle&\ldots & 0 &\ldots &\langle\p_i\,\p_{n-2}\rangle & -\langle y\,\p_i \rangle \\
         \langle \p_j\,\p_1\rangle & \ldots & 0 &\ldots & -\langle \p_j\,\p_i\rangle & \ldots & \langle \p_j\,\p_{n-2}\rangle & -\langle y\,\p_j\rangle 
        \end{pmatrix}\label{P-matrix}\,.
    \end{align}
    \end{subequations}
    \normalsize
The integration domain $\Delta_{n}\in \R^{2n}$ in \eqref{V-subvertices} is a compact convex $n$-gon and the integrating variables $u_{k}$ and $v_{k}$ obey 
\begin{align*}
    \frac{u_{1}}{v_{1}}\leq \frac{u_{2}}{v_{2}}\leq \ldots\leq \frac{u_{n}}{v_{n}}\leq \frac{1-U}{1-V}\,,\qquad U=\sum u_{i}\,,\quad  V=\sum v_{i}\,.
\end{align*}
Furthermore, $G_{ij}$ denotes the corresponding tree/graph in homological perturbation theory from which $V$'s come from. Additionally, $\sigma_{G_{ij}}$ represents the permutation that arranges fields into `natural' order of the associated graph. Lastly, 
\begin{align}\label{mho-n-factor}
    \mho_n:=\exp\Big(\sum_{i=1}^n[y\,\p_i]+\sum_{i< j\leq n}[\p_i\,\p_j]\Big)\,,
\end{align}
is a poly-differential operator acting solely on the $\tilde y$-sector by virtue of the $\star$-product \eqref{eq:star-product}.

Canonically, we can take $G_{12}=(\omega_1,\omega_2,\sC_3,\ldots,\sC_{n-2})$ to be the reference graph where all $\sC$'s are on the right of $\omega$'s. For concreteness, the explicit form of the sub-vertex $V^{(n)}_{\omega,\omega,\sC,\ldots,\sC}$ associated with $G_{12}$ has the following form \small
\begin{align}\label{V-subvertices-canonical}
        &V^{(n)}_{\omega_1,\omega_2,\sC,\ldots,\sC}(y,\tilde y;\p_i,\tilde \p_i)=\mho_n \langle \p_1\,\p_2\rangle^{n-2} \times\nn\\
        &\quad \times \int_{\Delta_{n-2}} \exp\Bigg[\Big(1-\sum_{k=1}^{n-2} u_{k}\Big)\langle y\,\p_1\rangle +\Big(1-\sum_{k=1}^{n-2}v_{k}\Big)\langle y\,\p_2\rangle+\sum_{m\neq 1,2}u_{m}\langle \p_1\,\p_{m}\rangle+ \sum_{m\neq 1,2}v_{m}\langle \p_2\,\p_{m}\rangle\Bigg]\,.
    \end{align}
    \normalsize
Then, any other graph can be obtained from $G_{12}$ by moving some number of $\sC$'s in between $\omega$'s and cyclicing all arguments to a specific position e.g. $V_{\sC,\ldots,\omega,\sC,\ldots,\omega,\sC,\sC}\rightarrow V_{\sC,\sC,\sC,\ldots,\omega,\sC,\ldots,\omega}$. This allows us to constructs the remaining vertices/graphs from the reference ones.

    $\diamond$ \underline{$\cU(\omega,\sC,\ldots,\sC)$ vertices.} It turns out that $(\cV,\cU)$ vertices can be related to each others through a $\C$-linear pairing inherited from a degree 2 cyclic $A_{\infty}$-algebra:
    $$\langle-|-\rangle:\cN\otimes\cN^{\vee}[-1]\rightarrow \C[-2]\,, \quad \text{where}\quad \langle a|b\rangle=(-)^{(|a|+1)(|b|+1)}\langle b|a\rangle\,.$$
    The natural pairings in the $y$-sector used to construct all vertices \cite{Sharapov:2022faa,Sharapov:2022awp,Sharapov:2022nps} are as follows:
    \begin{table}[ht!]
    \centering
    \begin{tabular}{|c|} \hline
       $\langle \omega|\sC\rangle=-\langle \sC|\omega\rangle=\exp[\langle \p_1\,\p_2\rangle]\omega(y_1)\sC(y_2)$  \\ \hline
       {$\!\begin{aligned}
     \big\langle V(\omega,\omega,\sC,\ldots,\sC)| C\big\rangle&=+\big\langle \omega|U(\omega,\sC,\ldots,\sC)\big\rangle\\
    V(\p_1,\ldots,\p_n|y)&\mapsto +U(y,\p_1,\ldots,\p_{n-1}|-\p_n)
       \end{aligned}$}\\ \hline
       {$\!\begin{aligned}
     \big\langle \sC| V(\sC,\ldots,\omega,\sC,\ldots,\sC,\omega)\big\rangle&=-\big\langle U(\sC,\sC,\ldots,\omega,\sC,\ldots,\sC)|\omega\big\rangle\nn\\
    V(y|\p_1,\ldots,\p_n)&\mapsto -U(-\p_1|\p_2,\ldots,\p_n,y)
       \end{aligned}$}\\\hline
    \end{tabular}
    \caption{Some natural pairings in chiral-FDA.}
    \label{tab:natural-pairing}
\end{table} \\  
For instance, to obtain $U_{\sC,\ldots,\sC,\omega,\ldots,\sC}\equiv U(\sC,\ldots,\sC,\omega,\ldots,\sC)$ from $V_{\omega,\ldots,\sC,\omega,\ldots,\sC}$, we replace $y\mapsto -\p_n$, $\p_1 \mapsto y$ and $\p_i\mapsto \p_{i-1}$ (for $i\neq n$), while to obtain $U_{\sC,\ldots,\omega,\ldots,\sC}$ from $V_{\sC,\ldots,\omega,\sC,\ldots,\sC,\omega}\equiv V(\sC,\ldots,\omega,\sC,\ldots,\sC,\omega)$, we replace $y\mapsto -\p_1$, $\p_i\mapsto \p_{i+1}$ (for $i\leq n$) and $\p_n\mapsto y$ with an overall minus sign. Since there are $\binom{n}{2}=\frac{n(n-2)}{2}$ number of $V$'s sub-vertices compared to $n$ number of $U$'s sub-vertices at a given order $n$, this replacement rule always works. Therefore, as a matter of convenience, one only needs to construct all $A_{\infty}$ structure maps $V^{(n)}$, and everything will follow.

Note that the pairing $\langle -|-\rangle $ not only facilitates the cyclic reordering of fields and vertices, but is also defined to eliminate all auxiliary $y$ and $\tilde y$ variables. This ensures that all possible contractions between the stand-alone field $\omega$ or $\sC$ 
with the multilinear maps or sub-vertices $V$'s and $U$'s are fully captured.


\section{Kerr-Schild ansatz and exact solutions}\label{sec:chiral-FDA}
Armed with the above chiral FDA data, we now propose a family of explicit representations of $(\omega,\sC)$ and show that they form exact harmonic solutions in chiral HSGRA. 

Recall that since the positive-helicity fields $C$'s can be generated from $\omega$ via linear chiral FDA, cf. \eqref{linear-FDA-general}, we can construct all positive-helicity data starting from the physical fields $\cA_{\alpha(2s-2)}$. Furthermore, we show below that $\cA_{\alpha(2)}$ is the field that generate a self-dual background. It is also worth stressing before running our analysis that $B^{\alpha(2s)}$ and their $\nabla$-generated thereof are \emph{not} derived objects from the positive helicity fields $\cA_{\alpha(2s-2)}$. Rather, they should be regarded simply as linearized fields that propagate on the self-dual background generated by $\cA_{\alpha(2)}\equiv\omega_{\alpha(2)}$.

\subsection{Cartan's structure equations for empty self-dual background}
One of the nice features of $\hs_{\Lambda}$ is that its $\mso(2,3)$ subalgebra can be defined via the generators:
\begin{align}\label{eq:generators}
    L_{\alpha\alpha}=-\frac{1}{2}y_{\alpha}y_{\alpha}\,,\qquad P_{\alpha\dot\alpha}=-\frac{1}{2}y_{\alpha}\tilde y_{\dot\alpha}\,,\qquad \widetilde L_{\dot\alpha\dot\alpha}=-\frac{1}{2}\tilde y_{\dot\alpha}\tilde y_{\dot\alpha}\,.
\end{align}
As usual, $P_{\alpha\dot\alpha}$ is translation generator, and $(L_{\alpha\alpha},\widetilde L_{\dot\alpha\dot\alpha})$ are Lorentz's. The algebra which describes AdS${}_4$ is \cite{Krasnov:2021cva} 
\begin{subequations}\label{eq:SDBG-algebra}
    \begin{align}
        [L_{\alpha\beta},L_{\gamma\delta}]_{\star}&= \Lambda\big(\epsilon_{\beta\gamma}L_{\alpha\gamma}+\epsilon_{\alpha\gamma}L_{\beta\delta}+\epsilon_{\beta\delta}L_{\alpha\gamma}+\epsilon_{\alpha\gamma}L_{\beta\delta}\big)\,,\\
        [\widetilde L_{\dot\alpha\dot\beta},\widetilde L_{\dot\gamma\dot\delta}]_{\star}&=\big(\epsilon_{\dot\beta\dot\gamma}L_{\dot\alpha\dot\gamma}+\epsilon_{\dot\alpha\dot\gamma}L_{\dot\beta\dot\delta}+\epsilon_{\dot\beta\dot\delta}L_{\dot\alpha\dot\gamma}+\epsilon_{\dot\alpha\dot\gamma}L_{\dot\beta\dot\delta}\big)\,,\\
        [L_{\alpha\beta},P_{\gamma\dot\gamma}]_{\star}&=\Lambda\Big(\epsilon_{\beta\gamma}P_{\alpha\dot\gamma}+\epsilon_{\alpha\gamma}P_{\beta\dot\gamma}\Big)\,,\\
        [\widetilde L_{\dot\alpha\dot\beta},P_{\gamma\dot\gamma}]_{\star}&=\big(\epsilon_{\dot\beta\dot\gamma}P_{\gamma\dot\alpha}+\epsilon_{\dot\alpha\dot\gamma}P_{\gamma\dot\beta}\big)\,,\\
        [P_{\alpha\dot\alpha},P_{\beta\dot\beta}]_{\star}&=\big(\epsilon_{\dot\alpha\dot\beta}L_{\alpha\beta}+\Lambda \epsilon_{\alpha\beta}\widetilde L_{\dot\alpha\dot\beta}\big)\,.
    \end{align}
\end{subequations}
It is easy to notice that $L_{\alpha\beta}$ becomes central of the algebra \eqref{eq:SDBG-algebra} in the flat limit $\Lambda\rightarrow 0$. That is
\begin{align}
    [P_{\alpha\dot\alpha},P_{\beta\dot\beta}]_{\star}\big|_{\Lambda=0}=\epsilon_{\dot\alpha\dot\beta}L_{\alpha\beta}\,,
\end{align}
instead of being trivial. For this reason, the algebra generated by \eqref{eq:generators} in the flat limit is referred to as a \emph{deformed} chiral Poincar\'e algebra \cite{Ponomarev:2022ryp}.

The Cartan structure equations associated to \emph{empty} flat spacetime read:
\begin{align}\label{Cartan1}
    d\Omega-\frac{1}{2}[\Omega,\Omega]=0\,,\qquad \Omega=\frac{1}{2}\varpi_{\alpha\beta}L^{\alpha\beta}+h_{\alpha\dot\alpha}P^{\alpha\dot\alpha}+\frac{1}{2}\widetilde\varpi_{\dot\alpha\dot\beta}\widetilde L^{\dot\alpha\dot\beta}\,.
\end{align}
where $h_{\alpha\dot\alpha}$ is the background vierbein and $\varpi_{\alpha\beta}$, $\widetilde\varpi_{\dot\alpha\dot\beta}$ are the (A)SD components of the spin connection. In terms of components, \eqref{Cartan1} unfolds into
\begin{subequations}\label{eq:evarpi1}
\begin{align}
    R_{\alpha\beta}:&=d\varpi_{\alpha\alpha}+H_{\alpha\alpha}=0\,,\label{Rab1}\\ 
    T_{\alpha\dot\alpha}:&=dh_{\alpha\dot\alpha}+\widetilde\varpi_{\dot\alpha\dot\gamma}\wedge h_{\alpha}{}^{\dot\gamma}\equiv \nabla h_{\alpha\dot\alpha}=0\,,\label{tor1}\\
    R_{\dot\alpha\dot\beta}:&=d\widetilde\varpi_{\dot\alpha\dot\alpha}+\widetilde\varpi_{\dot\alpha\dot\gamma}\wedge \widetilde\varpi_{\dot\alpha}{}^{\dot\gamma}=0\,. \label{Rab2} 
\end{align}
\end{subequations}
The reader may notice the difference with the Cartan structure equations associated with the conventional Poincar\'e algebra where $[P,P]=0$. In particular, for SD background, the ASD component of the spin connection does \emph{not} participate in the torsion constraint \eqref{tor1}. Furthermore, the Riemann two-form \eqref{Rab1} reveals an intriguing feature for SD background. That is, the spin connection $\varpi_{\alpha\alpha}$ is the main player of this game since the vierbein and the SD component of the spin connection can be solved for once we know $\varpi_{\alpha\alpha}.$

\subsection{Deformed self-dual pp-wave background}
This subsection proposes a class of SD pp-wave backgrounds in which chiral HSGRA is well-defined. This class of backgrounds is generated by a fluctuation of the ASD component of the spin connection, which is characterized by a function $\tH_2(v,z,\tilde z)$ that is harmonic in the complex 2-plane $(z,\tilde z)$, and two principal spinors $(\lambda,\tilde\lambda)$.\footnote{As is well-known, any symmetric tensor of the form $S_{\alpha(s)}\equiv S_{\alpha_1\ldots \alpha_s}$ can be factorized into a symmetrized product of one-index spinors \cite{Penrose:1986ca}, i.e.
\begin{align}\label{decompositionPS}
    S_{\alpha_1\ldots\alpha_s}=\tau_{(\alpha_1}\ldots\zeta_{\alpha_s)}\,,
\end{align}
where each spinor $\{\tau,\ldots,\zeta\}$ is referred to as
\emph{principal spinor} of $S_{\alpha(s)}$. Note that these principal spinors can be used to classify the Weyl tensors in classical GR, see e.g. \cite{Newman:1961qr,Penrose:1986ca}.} See also e.g. previous work \cite{Tod:1979tt,Sparling:1981nk} for the construction of self-dual pp-wave solution in somewhat similar vein. 
\paragraph{Deformed Cartan's structure equation.} Suppose we want to construct a non-trivial vacuum solution which does not describe an empty flat space. Then, the new vacuum solution whose geometrical information is encoded by the traceless part of the Riemann tensor, i.e. the Weyl tensor, can be regarded as a deformation away from empty flat space. In this case, we shall modify the Cartan's structure equations \eqref{eq:evarpi1} to
\begin{subequations}\label{eq:deform-Cartan}
\begin{align}
    \Rbold_{\alpha\alpha}:&=d\varpibold_{\alpha\alpha}+\Hbold_{\alpha\alpha}=0\,,\label{eq:Rbold1}\\
    D\ebold_{\alpha\dot\alpha}&=d\ebold_{\alpha\dot\alpha}+\widetilde\varpibold_{\dot\alpha\dot\beta}\wedge \ebold_{\alpha}{}^{\dot\beta}=0\,,\label{deformedvierbeinpostulate}\\
    \widetilde \Rbold_{\dot\alpha\dot\alpha}:&=d\widetilde\varpibold_{\dot\alpha\dot\beta}+\widetilde\varpibold_{\dot\alpha\dot\gamma}\wedge \widetilde\varpibold_{\dot\beta}{}^{\dot\gamma}=\widetilde \Hbold^{\dot\beta\dot\beta}C_{\dot\beta\dot\beta\dot\alpha\dot\alpha}\,,\label{eq:Rbold2}
\end{align}
\end{subequations}
where 
\begin{align}
    \quad \varpibold_{\alpha\beta}=\varpi_{\alpha\alpha}+\omega_{\alpha\alpha}\,,\quad \ebold_{\alpha \dot \alpha}=h_{\alpha\,\dot\alpha}+\omega_{\alpha\dot\alpha},\quad \widetilde\varpibold_{\dot\alpha\dot\beta}=\widetilde\varpi_{\dot\alpha\dot\beta}+\omega_{\dot\alpha\dot\beta}\,,
\end{align}
are the deformed vierbein and spin connections with $\omega_{\alpha\alpha},\omega_{\alpha\dot\alpha},\omega_{\dot\alpha\dot\alpha}$ being the associated fluctuations.\footnote{All the \emph{bold} symbols are geometrical data of the deformed background. } Here, $D$ stands for the connection of the deformed SD background, and \eqref{eq:deform-Cartan} contains only the SD component $C_{\dot\alpha(4)}$ of the Weyl tensor 
due to SD condition, cf. \cite{Krasnov:2021nsq}.

\paragraph{Tetrad and dyad bases.} 
In $4d$, we can associate to each point $x\in \cM$ a \emph{tetrad} of complex null vectors $\zeta^a=(\ell^a,n^a,p^a,\tilde p^a)$ (see e.g. \cite{Penrose:1986ca}) where
\begin{align}\label{eq:normalization}
    \ell_an^a=-1=p_a\tilde p^a\,,\qquad \qquad  \ell_ap^a=\ell_a \tilde p^a=n_ap^a=n_a\tilde p^a=0\,.
\end{align}
Then, a vector $V^a$ may be written in the tetrad basis as
\begin{align}
    V^a=u\ell^a+vn^a+zp^a+\tilde z \tilde p^a\,.
\end{align}
Here, we may refer to $(u,v,z,\tilde z)$ as spacetime coordinates. 

Since the local Lorentz group of a $4d$ complexified $\cM$ is
$SL(2,\C)\times SL(2,\C)$, we can also associate to each point $x\in \cM$ the \emph{dyad} bases $(\lambda^{\alpha},\tilde\lambda^{\dot\alpha})$ and $(\mu^{\alpha},\tilde\mu^{\dot\alpha})$, cf. \cite{Penrose:1986ca}, where $(\mu,\tilde\mu)$ is another pair of principal spinors such that
\begin{align}\label{eq:spinor-norm}
    \langle \mu\,\lambda\rangle=1=[\tilde\mu \,\tilde\lambda]\,.
\end{align}
Using the same normalization in \eqref{eq:normalization}, we can fix unambiguously
\begin{align}
    \ell_a= \lambda_{\alpha}\tilde\lambda_{\dot\alpha}\,,\qquad n_a=  \mu_{\alpha}\tilde\mu_{\dot\alpha}\,,\qquad p_a=\lambda_{\alpha}\tilde\mu_{\dot\alpha}\,,\qquad \tilde p_a=\mu_{\alpha}\tilde\lambda_{\dot\alpha}\,.
\end{align}
In constructing the proposed harmonic solutions for chiral HSGRA, we will use $(\lambda,\tilde\lambda)$ to characterize positive-helicity fields and $(\mu,\tilde\mu)$ to describe the negative-helicity fields as they are related in the sense of \eqref{eq:free action}.

Note that we will define the spinors $(\lambda,\tilde\lambda)$ as well as $(\mu,\tilde\mu)$ to be covariantly constant wrt. the connection associated to empty flat background, i.e. 
\begin{align}
    \nabla\lambda_{\alpha}=0\,,\qquad \nabla\tilde\lambda_{\dot\alpha}=0\,,\qquad \nabla\mu_{\alpha}=0\,,\qquad \nabla\tilde\mu_{\dot\alpha}=0\,.
\end{align}
\paragraph{Kerr-Schild ansatz.} Typically, in standard GR, one can simplify the process of finding exact solutions for Einstein equations by considering the Kerr-Schild ansatz \cite{kerr2009republication}
\begin{align}\label{eq:KerrSchildAnsatz}
    \gbold_{ab}=g_{ab}+\tH \ell_a\ell_b\,,\qquad \qquad \gbold^{ab}=g^{ab}-\tH \ell^a\ell^b\,
\end{align}
where 
$g_{ab}$ is the metric of the background and $\tH$ is a scalar profile function whose role is to generate a deformed spacetime $(\cM,\gbold)$. However, in SDGR, the situation is slightly different. In particular, as the ASD component of the spin connection is the main object where all other geometrical data are derived from, cf. \cite{Krasnov:2016emc}, one must consider a modified version of the Kerr-Schild ansatz, adapted to accommodate the ASD component of the spin connection. For our purpose of constructing exact solutions in SD setting, we can consider
\begin{align}\label{eq:chiral-ansatz-1}
    \varpibold_{\alpha\alpha}=\varpi_{\alpha\alpha}+\tH_2 \langle \lambda|h|\lambda] \lambda_{\alpha}\lambda_{\alpha}\,,\qquad \langle \lambda|h|\lambda]:=\lambda^{\gamma}h_{\gamma\dot\gamma}\tilde\lambda^{\dot\gamma}\,.
\end{align}
Note that the above ansatz is inspired from the recent work \cite{Adamo:2022lah,Neiman:2023bkq}. For the sake of generality, we have introduced a subscript ``2'' to the profile function $\tH$ to indicate that it corresponds to the ansatz for the spin-2 sector. For spin-$s$ sector, we may denote the corresponding profile function as $\tH_s$ (see below). 

\paragraph{Harmonicity of $\tH_2$.} There is a few properties that we want to impose on the profile function $\tH_2$. Namely, $\tH_2$ must be $u$-independent and harmonic in $(z,\tilde z)$-plane spanned by $(p^a,\tilde p^a)$ \cite{Stewart:1990uf}, i.e. 
\begin{align}\label{eq:harmonicity-H2}
    \ell^a\nabla_a\tH_2=0\,,\qquad p^a\tilde p^b\nabla_a\nabla_b\tH_2=0\,.
\end{align}
This allows us to write $\tH_2(v,z,\tilde z)$ as the sum of two functions
\begin{align}
    \tH_2(v,z,\tilde z)=\tH_2(v, z)+\widetilde\tH_2(v,\tilde z)\,,
\end{align}
so that $\tH_2$ effectively behaves as a $2d$ off-shell scalar field. Nevertheless, as shown below, the self-duality condition will set one of them, specifically $\widetilde{\tH}_2(v, \tilde z)$, to zero. 

The harmonicity of $\tH_2$ can be recast into the following pde: 
\begin{align}\label{eq:harmonic-fda-2}
    \nabla\tH_2=h^{\alpha\dot\alpha}(\lambda_{\alpha}\tilde\tH_{2\,\dot\alpha}+\tilde\lambda_{\dot\alpha}\tH_{2\,\alpha})\,,
\end{align}
where $\tH_{2\,\alpha}$ and $\tilde\tH_{2\,\dot\alpha}$ are first-order derivatives of $\tH_2$ in the $y$- and $\tilde y$-sector, respectively. Furthermore, all derivatives of $\tH_2$, i.e. $\nabla^{\alpha\dot\alpha}\ldots \nabla^{\alpha\dot\alpha} \tH_2 \equiv \tH_2^{\alpha(k)\,\dot\alpha(k)}$ can be factorized as
\begin{align}\label{fractorize}
    \tH^{\alpha(k)\,\dot\alpha(k)}&= \lambda^{\alpha(k)} \widetilde\tH_2^{\dot\alpha (k)} +\tilde\lambda^{\dot\alpha(k)}\tH_2^{\alpha(k)}\,.
\end{align}
This special property of $\tH_2$ can be used to generate the linear FDA for the spin-2 sector. However, let us first specify the deformed data in \eqref{eq:deform-Cartan}.
\paragraph{Solving deformed Cartan's structure equations.} Assuming $C_{\dot\alpha(4)}=\phi_{(0,4)}\tilde\lambda_{\dot\alpha(4)}$, where $\phi_{(0,4)}$ is a scalar function, cf. \eqref{eq:phi(0,4)}, and plugging the ansatz \eqref{eq:chiral-ansatz-1} to \eqref{eq:deform-Cartan}, we obtain
\begin{subequations}\label{eq:SDGR-data}
    \begin{align}
        [\lambda\,\tH_2]&=0\,,\label{eq:SD-tH}\\
        \varpibold_{\alpha\alpha}&=\varpi_{\alpha\alpha}+\tH_2 \langle \lambda|h|\lambda] \lambda_{\alpha}\lambda_{\alpha}\,,\\
        \ebold_{\alpha\dot\alpha}&=h_{\alpha\dot\alpha}+\langle \lambda\,\tH_2\rangle\langle \lambda|h|\lambda]\lambda_{\alpha}\tilde\lambda_{\dot\alpha}\,,\label{eq:deformed-veirbein}\\
        \widetilde\varpibold_{\dot\alpha\dot\beta}&=\widetilde\varpi_{\dot\alpha\dot\beta}+\langle \lambda^{\gamma(2)}\,\tH_{2\,\gamma(2)}\rangle \langle \lambda|h|\lambda]\tilde\lambda_{\dot\alpha}\tilde\lambda_{\dot\alpha}\,,\\
        C_{\dot\alpha(4)}&=\phi_{(0,4)}\tilde\lambda_{\dot\alpha(4)}\,,\qquad \phi_{(0,4)}=\frac{1}{2}\langle \lambda^{\gamma(3)}\,\tH_{2\,\gamma(3)}\rangle \,,\label{eq:phi(0,4)}
    \end{align}
\end{subequations}
where
\begin{align}
    [\lambda\,\tH_2]:=\tilde\lambda^{\dot\alpha}\,\tH_{2\,\dot\alpha}\,,\qquad \langle \lambda^{\alpha(n)}\,\tH_{2\,,\alpha(n)}\rangle:= \lambda^{\alpha(n)}\,\tH_{2\,\alpha(n)}\,,\qquad \lambda^{\alpha(n)}\equiv\lambda^{\alpha_1}\ldots \lambda^{\alpha_n}\,.
\end{align}
In deriving the above results, the following relations are useful:
\begin{subequations}
    \begin{align}
        h_{\alpha\dot\alpha}\langle \lambda|h|\lambda]&=-\frac{1}{2}\Big(H_{\alpha\beta}\lambda^{\beta}\tilde\lambda_{\dot\alpha}+\widetilde H_{\dot\alpha\dot\beta}\tilde\lambda^{\dot\beta}\lambda_{\alpha}\Big)\,,\\
        d\langle \lambda|h|\lambda]\,\tH_2&=+\frac{1}{2}\Big(\cH[\lambda\,\tH_2]+\widetilde\cH \langle \lambda\,\tH_2\rangle\Big)\,,\\
        d\langle \lambda|h|\lambda]\langle \lambda^{\alpha(k)}\,\tH_{2\,\alpha(k)}\rangle&=+\frac{1}{2}\widetilde\cH\langle \lambda^{\alpha(k+1)}\,\tH_{2\,\alpha(k+1)}\rangle\,,
    \end{align}
\end{subequations}
where 
\begin{align}
    \cH:=H_{\alpha(2)}\lambda^{\alpha(2)}\,,\quad \widetilde\cH:=\widetilde H_{\dot\alpha(2)}\tilde\lambda^{\dot\alpha(2)}\,.
\end{align}
Observe that \eqref{eq:SD-tH} implies the vanishing of all higher-order derivatives of $\tH_{2}$ in the $\tilde y$-sector. Thus, our analysis will focus solely on the derivative of $\tH_2$ in the $y$-sector.

From \eqref{eq:SDGR-data}, it is easy to see that the self-dual connection takes the form 
\begin{align}\label{eq:SD-connection}
    D=\nabla-\frac{1}{2}\tH_2\langle \lambda|h|\lambda]\lambda_{\alpha}\lambda_{\alpha}L^{\alpha\alpha}-\frac{1}{2}\langle \lambda^{\gamma(2)}\,\tH_{2\,\gamma(2)}\rangle \langle \lambda|h|\lambda]\tilde\lambda_{\dot\alpha}\tilde\lambda_{\dot\alpha}\widetilde L^{\dot\alpha\dot\alpha}\,.
\end{align}
Remarkably, the principle spinors $\lambda_{\alpha},\tilde\lambda_{\dot\alpha}$ are again covariantly constant wrt. the deformed connection $D$, i.e.
\begin{align}\label{eq:covariant-constancy-lambda}
    D\lambda_{\alpha}=0\,,\qquad D\tilde\lambda_{\dot\alpha}=0\,.
\end{align}
Therefore, $D$ only acts non-trivially on the profile function $\tH_2$ and its derivative thereof, while ``effortlessly'' passing through $(\lambda,\tilde\lambda)$. This fact is essential for us to obtain a simple exact solution for chiral HSGRA with the above setting.

Note, however, that
\begin{align}\label{eq:non-covariant-constancy-mu}
    D\mu^{\alpha}=-\tH_2\langle \lambda|h|\lambda]\lambda^{\alpha}\,,\qquad D^2\mu^{\alpha}=0\,.
\end{align}
Thus, unlike $(\lambda,\tilde\lambda)$, the principal spinors $(\mu,\tilde\mu)$ are not covariantly constant wrt. $D$. 
\paragraph{Metric of self-dual pp-wave background.} Remarkably, the SD background metric generated by $\omega_{\alpha\alpha}$ only slightly deviates from the conventional solutions in GR, cf. \cite{brinkmann1925einstein,Aichelburg:1970dh,bilge1983generalized,Stewart:1990uf,Stephani:2003tm}, by being self-dual and well-defined in Euclidean or split signature. 

$-$ Recall that the line element for Brinkmann metric in the conventional GR is 
\begin{align}
    d s^2=-du\,d v+\tH(v,z,\bar z) d v^2+ d zd\bar z\,.
\end{align}
Here, $\tH$ is the profile function, which is harmonic in $(z,\bar{z})$-plane, 
that generates a generic Petrov type-N spacetime given that $\ell_a=\lambda_{\alpha}\tilde\lambda_{\dot\alpha}$. As $\tH(v,z,\bar{z})$ is harmonic in the $(z,\bar z)$ plane, we can write
\begin{align}\label{generalH}
    \tH(v,z,\bar z)=f_1(v,z)+f_2(v,\bar{z})\,,
\end{align}
for generic functions $f_1(v,z)$, $f_2(v,\bar z)$. For real solutions, then it is necessary that $(f_2)^*=f_1$. 

$-$ In SD setting, one can deduce from the deformed vierbein \eqref{eq:deformed-veirbein} that
\begin{align}
    g^{SD}_{\mu\nu}=\eta_{\mu\nu}+\langle\lambda\,\tH_2\rangle\ell_\mu\ell_{\nu}\,.
\end{align}
That is, the metric is now generated by the first-derivative of $\tH_2$ in the $z$-direction.\footnote{The derivative is represented by the operator $p^a\nabla_a$.} In such a setting, we have
\begin{subequations}\label{eq:deform-metric-SD}
    \begin{align}
    ds^2_{SD}&=-du\,dv+\langle \lambda\,\tH_2\rangle dv^2+dzd\tilde z\,,\qquad \langle\lambda\,\tH_2\rangle=j(v,z)\,,\\
    \phi_{(0,4)}&=\langle \lambda^{\alpha(3)}\,\tH_{2\,\alpha(3)}\rangle \sim \p_{z}^3j(v, z)\,.
\end{align}
\end{subequations}
That is, the Brinkmann-like profile function is now dependent solely on $(v, z)$, consistent with the self-duality of the solution.

Intriguingly, even in the SD case, it is possible to place multiple pp-waves at different locations $v_i$ along the $v$-direction so that the total profile function for the self-dual pp-wave metric can be written as $\langle \lambda\,\tH_2\rangle =\sum_i\langle \lambda\,\tH_2^i(v_i)\rangle$.


\subsection{Linear chiral-FDA}
This subsection studies the linearized FDA for chiral HSGRA. Let us begin by stating our Kerr-Schild ansatz proposal for all positive-helicity spinning fields. 

\subsubsection{Kerr-Schild ansatz for all spin}

Following the discussion around \eqref{eq:simplify-omega}, it is reasonable to consider the following ansatz for all spins
\begin{subequations}\label{eq:LC-ansatz}
    \begin{align}
    \omega_{\alpha(2s-2)}:&=\langle \lambda|h|\lambda]\lambda_{\alpha(2s-2)}\tH_s\,,\label{omega-H}\\
    B^{\alpha(2s)}&:=\psi_s\mu^{\alpha(2s)}\,,\qquad \psi_s:=\lambda_{\alpha(2s)}B^{\alpha(2s)}\label{psi-s}\,,\quad \langle\mu\,\lambda\rangle=1\,,
\end{align}
\end{subequations}
where $\tH_s$ is a harmonic scalar function associated with a positive-helicity massless spin-$s$ field, i.e. 
\begin{align}\label{eq:harmonic-s}
    \ell^a\nabla_a\tH_s=0\,,\qquad p^a\tilde p^b\nabla_a\nabla_b\tH_s=0\,,
\end{align}
as in the case of $\tH_2$, cf. \eqref{eq:harmonicity-H2}. Furthermore, for the negative-helicity fields $B^{\alpha(2s)}$, we will assume that
\begin{align}\label{eq:harmonicity-psi-s}
    \nabla \mu^{\alpha}=0\,,\qquad \ell^a\nabla_a\psi_s=0\,,\qquad p^a\tilde p^b\nabla_a\nabla_b\psi_s=0\,.
\end{align}
This guarantees that the free action of the model is \eqref{eq:free action}. Note that on the proposed SD pp-wave background
\begin{align}
    D_{\beta}{}^{\dot\alpha}B^{\beta\alpha(2s-1)}=0\,,
\end{align}
if $p^a\nabla_a\psi_s=0$, or $\psi=\psi(v,\tilde z)$. This is in accordance with the fact that $\psi_s$ are harmonic, cf. \eqref{eq:harmonicity-psi-s}, and $B^{\alpha(2s)}$ are linear fluctuations on the proposed self-dual background. Thus, \eqref{eq:harmonicity-psi-s} can serve as a suitable ansatz for the negative-helicity fields.

\subsubsection{\texorpdfstring{$\tH$}{H}-FDA}
Let us now cast the system \eqref{eq:harmonic-fda-2} into generating function form. Since we have assumed $\tH_s$ to be harmonic, cf. \eqref{eq:harmonic-s}, all $\tH_s^{\alpha(k)\,\dot\alpha(k)}$ will admit the same factorization form as in \eqref{fractorize}. Then, we can write the factorization condition for $\tH_s^{\alpha(k)\,\dot\alpha(k)}$ using some generating functions $(\tH_s(y),\widetilde\tH_s(\tilde y))$ as 
\begin{subequations}
    \begin{align}
    \nabla \tH_s(y)&=\theta^{\alpha}\p_{\alpha}\tH_s(y)\,,\qquad \theta^{\alpha}:=\tilde\lambda_{\dot\beta}h^{\alpha\dot\beta}\,,\\
    \nabla\widetilde{\tH}_s(y)&=\tilde\theta^{\dot\alpha}\p_{\dot\alpha}\widetilde{\tH}_s(\tilde y)\,,\qquad \tilde\theta^{\dot\alpha}:=\lambda_{\beta}h^{\beta\dot\alpha}\,.
\end{align}
\end{subequations}
One can check that the above system satisfies the Bianchi identities
\begin{align}
    \nabla^2\tH_s(y)=0\,,\qquad \nabla^2\widetilde{\tH}_s(\tilde y)=0\,.
\end{align}
On the deformed background sourced by $\omega_{\alpha\alpha}$, cf. \eqref{eq:chiral-ansatz-1}, it can be shown that
\begin{subequations}\label{eq:H-FDA-general}
    \begin{align}
    D\tH_s^{\alpha(k)}&=\theta_{\beta}\tH_s^{\beta \alpha(k)}+k\tH_2\langle \lambda|h|\lambda]\lambda^{\alpha}\lambda_{\gamma}\tH_s^{\gamma\alpha(k-1)}\,,\\
    D\widetilde\tH_s^{\dot\alpha(k)}&=\tilde\theta_{\dot\beta}\widetilde\tH_s^{\dot\beta\dot\alpha(k)}+k\tH_2\langle\lambda^{\gamma(2)}\tH_{2\,\gamma(2)}\rangle\langle \lambda|h|\lambda]\tilde\lambda^{\dot\alpha}\tilde\lambda_{\dot\gamma}\widetilde\tH_s^{\dot\gamma\dot\alpha(k-1)}\,,
\end{align}
\end{subequations}
or in terms of generating function
\begin{subequations}
    \begin{align}
        D\tH_s(y)&=\theta^{\alpha}\big(\p_{\alpha}+\lambda_{\alpha}\langle\lambda\,y\rangle\langle \lambda\,\p\rangle\big)\tH_s(y)\,,\nn\\
        D\widetilde\tH_s(\tilde y)&=\tilde\theta^{\dot\alpha}\big(\tilde\p_{\dot\alpha}+\tilde\lambda_{\dot\alpha}\langle \lambda^{\gamma(2)}\tH_{2\,\gamma(2)}\rangle[\lambda\,y][\lambda\,\p]\big)\widetilde\tH_s(\tilde y)\,.
    \end{align}
\end{subequations}
Note that 
\begin{align}
    D^2\tH_s(y)=0\,,\qquad D^2\widetilde\tH_s(y)=0\,.
\end{align}
Thus, as long as $\tH_s$ is harmonic, cf. \eqref{eq:harmonic-s}, the above system can generate all auxiliary scalars which are derivative of $\tH_s$. 
\subsubsection{Linear FDA for spin-2 sector}
Let us now consider the FDA for the harmonic spin-2 solution. 
As observed earlier in \eqref{eq:SD-tH}, where $[\lambda\,\tH_2]=0$ due to self-duality constraint, it is natural to impose
\begin{align}
    \widetilde\tH_2^{\dot\alpha(k\geq 1)}=0\,,
\end{align}
which turns out to be a consistent condition for constructing a SD pp-wave background. 

In the spin-2 sector, all data we need to generate to complete a linear-FDA are
\begin{align}\label{eq:linear-FDA-spin-2-sector}
    \{\omega_{\alpha\alpha},\omega_{\alpha\dot\alpha},\omega_{\dot\alpha\dot\alpha}\}\bigcup_{k=0}^{\infty}\{C^{\alpha(k)\,\dot\alpha(k+4)}\}\bigcup_{m=0}^{\infty}\{B^{\alpha(m+4)\,\dot\alpha(m)}\}\,.
\end{align}
Using \eqref{eq:H-FDA-general}, and covariant constancy of $(\lambda,\tilde\lambda)$ wrt. $D$, cf. \eqref{eq:covariant-constancy-lambda}, we obtain
\begin{align}
    C^{\alpha(k)\,\dot\alpha(4+k)}=\frac{1}{2}\langle \tH_{2}^{\alpha(3+k)}\lambda_{\alpha(3)}\rangle \tilde\lambda^{\dot\alpha(4+k)}\,
\end{align}
from solving iteratively $\nabla C_{\alpha(k)\,\dot\alpha(4+k)}=h^{\beta\dot\beta}C_{\alpha(k)\beta\,\dot\alpha(4+k)\dot\beta}$. The solution can be written in terms of generating function as
\begin{align}
    C_4(y)=\frac{1}{(4+N_y)!}[y\,\lambda]^{4+N_y}\langle\lambda\,\p'\rangle^3 \tH_2(y')\Big|_{y'=y}\,,\qquad [y\,\lambda]=\tilde y^{\dot\alpha}\tilde\lambda_{\dot\alpha}\,,\quad N_y=y^{\alpha}\p_{\alpha}\,.
\end{align}
For the negative-helicity auxiliary fields $B^{\alpha(m+4)\,\dot\alpha(m)}$, we observe that $B^{\alpha(4+m)\,\dot\alpha(m)}$ for $m\geq 1$ is non-linear in $\tH_2$ and their derivatives even for the case $\psi_2=const$. This is due to the fact that $D\mu^{\alpha}\neq 0$ when the deformation sourced by $\omega_{\alpha\alpha}$ occurs. Although we cannot find a closed form for $B^{\alpha(m+4)\,\dot\alpha(m)}$, we can still prove that the set of linearized data \eqref{eq:linear-FDA-spin-2-sector} can solve the non-linear EOMs for chiral HSGRA, cf. \eqref{eq:FDA}. Therefore, they provide us exact solutions for the spin-2 sector.

\subsubsection{Linear FDA for all spins}
Proceeding analogously with the case of spin-2, we will assume that 
\begin{align}
    \widetilde\tH_s^{\dot\alpha(k\geq 1)}=0\,,
\end{align}
and find the following linear data
\begin{align}\label{eq:linear-FDA-spin-s-sector}
   \bigcup_s\left\{ \bigcup_{j=0}^{2s-2}\{\omega_{\alpha(2s-2-j)\,\dot\alpha(j)}\}\bigcup_{k=0}^{\infty}\{C^{\alpha(k)\,\dot\alpha(2s+k)}\}\bigcup_{m=0}^{\infty}\{B^{\alpha(2s+m)\,\dot\alpha(m)}\}\right\}\,.
\end{align}
\paragraph{one-form sector.} In analogy with the spin-2 sector, we start with the Kerr-Schild ansatz \eqref{omega-H}, solve iteratively the FDA equation $\nabla\omega_{\alpha(2s-2-k)\,\dot\alpha(k)}=h_{\alpha}{}^{\beta}\omega_{\alpha(2s-3-k)\,\dot\alpha(k)\dot\beta}$ and obtain
\begin{align}
    \omega_{\alpha(2s-2-k)\,\dot\alpha(k)}=\langle \lambda^{\gamma(k)}\,\tH_{s\,\gamma(k)}\rangle\langle \lambda|h|\lambda]\lambda_{\alpha(2s-2-k)}\tilde\lambda_{\dot\alpha(k)}\,,\qquad 0\leq k\leq 2s-2\,.
\end{align}
This can be written in terms of generating function as
\begin{align}\label{eq:omega-s-generating}
    \omega_s=\frac{\langle \lambda|h|\lambda]}{\Gamma[2s-1]}\Big(\langle y\,\lambda\rangle+[y\,\lambda]\langle\lambda\,\p'\rangle\Big)^{2s-2}\tH_s(y')\Big|_{y'=0}\,.
\end{align}
Here, we keep $\tH_{s\geq 3}\neq \tH_2$ for generality. However, it is intriguing noting that when $\{\tH_s\}=\tH_2$ for $\forall s$, we can subsume the full generating function $\omega$ of the one-form sector as
\begin{align}\label{eq:omegas-generating}
    \omega\big|_{\tH_s=\tH_2}=\sum_s\omega_{s\geq 1}\big|_{\tH_s=\tH_2}=\cosh\Big(\langle y\,\lambda\rangle+[y\,\lambda]\langle\lambda\,\p'\rangle\Big)\tH_2(y')\Big|_{y'=0}\,.
\end{align}
\paragraph{zero-form sector.} Similarly to the spin-2 case, we obtain the generating function for a spin-$s$ 0-form field $C$ as
\begin{align}
    C^{\alpha(k)\,\dot\alpha(2s+k)}=\frac{1}{2}\langle \tH_{2}^{\alpha(2s-1+k)}\lambda_{\alpha(2s-1)}\rangle \tilde\lambda^{\dot\alpha(2s+k)}\,
\end{align}
or
\begin{align}\label{eq:Cs-generating}
    C_s(y)=\frac{1}{(2s+N_y)!}[y\,\lambda]^{2s+N_y}\langle\lambda\,\p'\rangle^{2s-1} \tH_s(y')\Big|_{y'=y}\,,\qquad [y\,\lambda]=\tilde y^{\dot\alpha}\tilde\lambda_{\dot\alpha}\,,\quad N_y=y^{\alpha}\p_{\alpha}\,.
\end{align}
As before, $B^{\alpha(2s+m)\,\dot\alpha(m)}$ for $m\geq 1$ is non-linear in $\tH_2$ and their derivatives. However, we observe in general that
\begin{align}\label{eq:B-generating}
    B_s=B_s(\langle y\,\lambda\rangle,[y\,\lambda],\langle y\,\mu\rangle,\tH_2,\psi_s)\,.
\end{align}
This fact is sufficient to prove that the set in \eqref{eq:linear-FDA-spin-s-sector}, whose 1-form subset consists of $\omega_s$ generating functions (cf. \eqref{eq:omega-s-generating}), and the 0-form subset consists of $\sC_s=(C_s,B_s)$ fields, indeed solves the full non-linear equations of motion of chiral HSGRA.
\paragraph{Plane-wave solutions.} Let us consider the following profile function \cite{Skvortsov:2022syz}
\begin{align}
    \tH(y)=\exp\Big(\pm x^{\alpha\dot\alpha}\kappa_{\alpha}\tilde\lambda_{\dot\alpha}+\langle y\,\kappa\rangle\Big)\,,\qquad \langle\lambda\,\kappa\rangle=1\,,
\end{align}
where $\kappa$ is a spinor associated with the momentum of a plane-wave. Then, it can be computed from \eqref{eq:Cs-generating} and \eqref{eq:omega-s-generating} that
\begin{subequations}
    \begin{align}
        C&\sim \exp\Big(\pm x^{\alpha\dot\alpha}\kappa_{\alpha}\tilde\lambda_{\dot\alpha}+[y\,\lambda]\Big)\,,\\
        \omega&\sim \langle\lambda|h|\lambda]\exp\Big(\pm x^{\alpha\dot\alpha}\kappa_{\alpha}\tilde\lambda_{\dot\alpha}+\langle y\,\lambda\rangle+[y\,\lambda]\Big)\,.
    \end{align}
\end{subequations}
These represent a family of plane-wave solutions for chiral HSGRA.
\subsection{Non-linear chiral-FDA}
Typically, when finding solutions for the non-linear equations of motion \eqref{eq:FDA}, one often consider the following expansion of fields:
\begin{subequations}
    \begin{align}
        \omega&=\omega^{(1)}+\omega^{(2)}+\omega^{(3)}+\ldots=\Omegabold+\underline{\omega}+\omega^{(3)}\ldots\,,\\
        \sC&=\sC^{(1)}+\sC^{(2)}+\sC^{(3)}+\ldots=\underline{\sC}+\sC^{(2)}+\sC^{(3)}+\ldots\,,
    \end{align}
\end{subequations}
where $ \Omegabold=\frac{1}{2}\varpibold_{\alpha\beta}L^{\alpha\beta}+\ebold_{\alpha\dot\alpha}P^{\alpha\dot\alpha}+\frac{1}{2}\widetilde\varpibold_{\dot\alpha\dot\beta}\widetilde L^{\dot\alpha\dot\beta},$ and solve \eqref{eq:FDA} order by order as in e.g. \cite{Didenko:2021vui,Didenko:2021vdb}. However, it is not always the case one can solve for $(\omega,\sC)$ to all orders and sum them up, since the computations quickly become quite involved. 


Remarkably, it turns out that the $0$-form $\sC=(C,B)$ in \eqref{eq:Cs-generating} and the $1$-form $\omega$ in \eqref{eq:omegas-generating} solve the full non-linear equations of motion of chiral HSGRA. Thus, they form a family of exact solutions in chiral HSGRA. Note that we have denoted $\omega^{(2)}\equiv \underline{\omega}=\eqref{eq:omega-s-generating}$ as well as $\sC^{(1)}=\underline{\sC}=(\underline{C},\underline{B})$, cf. \eqref{eq:Cs-generating}.

\paragraph{Free equations of motion.} As a consistency check, let us first derive free EOMs from \eqref{eq:FDA}. Observe that for any generating function $f(y,\tilde y)$, we have
\begin{align}
    f\star\Omegabold\Big|_{\Lambda=0}=(1+[\p_1\,\p_2])e^{\langle y\,\p_1\rangle+\langle y\,\p_2\rangle+[y\,\p_1]+[y\,\p_2]}f(y_1,\tilde y_1)\Omegabold(y_2,\tilde y_2)\Big|_{\substack{y_i=0\\ \tilde y_i=0}}\,.
\end{align}
Thus, at linear order in fluctuation, we have
\begin{subequations}\label{eq:free-FDA}
    \begin{align}
        d\uomega&=\cV(\Omegabold,\uomega)+\cV(\Omegabold,\Omegabold,\usC)\,,\\
        d\usC&=\cU(\Omegabold,\usC)\,,
    \end{align}
\end{subequations}
where $\cU(\omega,\usC)=U_{\omega,\usC}+U_{\usC,\omega}$ and
\begin{subequations}\label{sub-U2}
    \begin{align}
        U_{\omega,\usC}&=+\frac{1}{2}\mho_2 \exp\Big(\langle y\,\p_2\rangle + \langle \p_1\,\p_2\rangle \Big)\omega_1(y_1,\tilde y_1)\usC_2(y_2,\tilde y_2)\,,\\
        U_{\usC,\omega}&=-\frac{1}{2}\mho_2 \exp\Big(\langle y\,\p_1\rangle - \langle \p_1\,\p_2\rangle \Big)\usC_1(y_1,\tilde y_1)\omega_2(y_2,\tilde y_2)\,,
    \end{align}
\end{subequations}
can be obtained from the natural pairings,\footnote{recall that $\mho_2$ is defined in \eqref{mho-n-factor}.} cf. Table \ref{tab:natural-pairing}. It can be then checked that \eqref{eq:free-FDA} reproduces \eqref{linear-FDA-general}.

\paragraph{$(\uomega,\usC)$ are exact solutions.} It is useful noticing that as $\cV$'s and $\cU$'s vertices are constructed from the set of $\left\{\langle y\,\p_i\rangle\,,\langle\p_i\,\p_j\rangle\,,[y\,\p_i]\,,[\p_i\,\p_j]\right\}$ operators, the building blocks for fields will not change, i.e. fields are always functions of $\{\langle y\,\lambda\rangle,[y\,\lambda],\langle y\,\mu\rangle,\langle \lambda\,\p'\rangle,\tH_s,\psi_s\}$. 


\begin{proposition}\label{prop:eraseV(Omega,C)} Let $$\Omegabold \quad \text{and}\quad  \usC=\usC(\langle y\,\lambda\rangle,\langle y\,\mu\rangle,[y\,\lambda],N_y,\langle \lambda\,\p'\rangle,\tH_2,\psi_2)$$ be the field content of the spin-2 sector. Then, all higher-order vertices $\cV_{n\geq 4}(\Omegabold,\Omegabold,\usC,\ldots,\usC)$ and $\cU_{n\geq 3}(\Omegabold,\usC,\ldots,\usC)$ vanish. Thus, the spin-2 sector characterized by $(\Omegabold,\usC)$ forms a family of exact solutions in chiral HSGRA.
\end{proposition}

\begin{proof} Let us first consider $\cV_{n\geq 4}(\Omegabold,\Omegabold,\usC,\ldots,\usC)$. From \eqref{V-subvertices}, it is easy to notice that there are \emph{forced} contractions between the background connection 1-form $\Omegabold$'s in the $y$-sector via the operations $\langle \p_i\,\p_j\rangle^{n-2}$ with $n$ being the number of fields entering the sub-vertices $V$. It is easy to notice that $\langle \p_i\,\p_j\rangle^{n-2}\Omegabold_i\Omegabold_j=0$ whenever $n>4$. For the case $n=4$, the contraction $\langle \p_i\,\p_j\rangle^2$ give\footnote{The $\wedge$-product is suppressed temporally for simplicity.}
\begin{align}
    \varpibold^{\alpha\alpha}\varpibold_{\alpha\alpha}=(\varpi+\omega)^{\alpha\alpha}(\varpi+\omega)_{\alpha\alpha}=0\,.
\end{align}
Therefore, $\cV_{n\geq 4}(\Omegabold,\Omegabold,\usC,\ldots,\usC)=0$. Note that the cubic vertex $\cV(\Omegabold,\Omegabold,\usC)$ is needed for consistency as it produces the SD components of the Weyl tensor cf. \eqref{eq:Rbold2}. 

Next, we turn our focus to $\cU_{n\geq 3}(\Omegabold,\usC,\ldots,\usC)$ vertices, which can be obtained via the $\cU$-$\cV$ duality maps, cf. \eqref{UV-duality-maps}. For $n>5$, each sub-vertex $U$ consists of a forced contraction of type $\langle y\,\p_i\rangle^{n-2}$ where $i$ is the position of $\Omegabold$ in $U_{\usC,\ldots,\usC,\Omegabold_i,\ldots,\usC}$. These sub-vertices clearly vanish for $n\geq 5$. Therefore, we only need to pay special attention to the cases where $n=3,4$. 

$\bullet$ For $n=3$, we have $\cU(\Omegabold,\usC,\usC)=U_{\Omegabold,\usC,\usC}+U_{\usC,\Omegabold,\usC}+U_{\usC,\usC,\Omegabold}$ where \cite{Sharapov:2022nps}
\begin{subequations}\label{U3}
\begin{align}
    U_{\Omegabold,\usC,\usC}&=+\mho_3\langle y \,\p_1\rangle\int e^{+(1-u)\langle y \,\p_3\rangle +(1-v)\langle \p_1 \,\p_3\rangle +u\langle y\,\p_2\rangle+v\langle \p_1\,\p_2\rangle }\,,\\
    U_{\usC,\Omegabold,\usC}&=-\mho_3\langle y\,\p_2\rangle \int e^{-(1-u)\langle  \p_1\,\p_2\rangle +(1-v)\langle y\, \p_1\rangle +u\langle \p_2\,\p_3\rangle +v\langle y\,\p_3\rangle}+ (v\leftrightarrow u)\,,\\
    U_{\usC,\usC,\Omegabold}&=+\mho_3\langle y\,\p_3\rangle \int e^{-(1-v)\langle \p_1\, \p_3\rangle +(1-u)\langle y \,\p_1\rangle -v\langle \p_2\,\p_3\rangle +u\langle y\,\p_2\rangle}\,.
\end{align}
\end{subequations}
Here, $\mho_3=\exp\big([y\,\p_1]+[y\,\p_2]+[y\,\p_3]+[\p_1\,\p_2]+[\p_1\,\p_3]+[\p_2\,\p_3]\big)$ (cf., \eqref{mho-n-factor}). We notice that there will be no contraction between $\usC$'s in the $\tilde y$ sector because the $\star$-product \eqref{eq:star-product} acts trivially on $\usC_i$ due to contractions between the principal spinors $\tilde\lambda$'s.\footnote{This is how one imposes locality into the formally consistent FDA \eqref{eq:FDA} making it becomes a non-trivial system that describes chiral HSGRA, cf. \cite{Sharapov:2022faa,Sharapov:2022awp,Sharapov:2022nps}.} Note that the prefactors $\langle y\, \p_i\rangle $ in each of the $U$ sub-vertices act on $\Omegabold$, so only $\varpibold_{\alpha\alpha}$ and $\ebold_{\alpha\dot\alpha}$ will give non-trivial contributions. 

 $-$ \underline{$\varpibold$'s contributions.} In this case, $\mho_3$ acts trivially. Thus, all $U$'s vertices reduce to
    \begin{subequations}
    \begin{align*}
    U_{\varpibold,\usC,\usC}&=+\int e^{+(1-u)\langle y\, \p_3\rangle +(1-v)\langle \p_1 \,\p_3\rangle +u\langle y\,\p_2\rangle+v\langle \p_1\,\p_2\rangle }(\varpibold_{1\alpha\alpha}y^{\alpha}y_1^{\alpha},\usC_2(y_2,\tilde y),\usC_3(y_3,\tilde y))\,,\\
    U_{\usC,\varpibold,\usC}&=-\int e^{-(1-u)\langle  \p_1\,\p_2\rangle +(1-v)\langle y \,\p_1\rangle +u\langle \p_2\,\p_3\rangle +v\langle y\,\p_3\rangle}(\usC_1(y_1,\tilde y),\varpibold_{2\alpha\alpha}y^{\alpha}y_2^{\alpha},\usC_3(y_3,\tilde y))\\
    &\quad -\int e^{-(1-v)\langle  \p_1\,\p_2\rangle +(1-u)\langle y\, \p_1\rangle +v\langle \p_2\,\p_3\rangle +u\langle y\,\p_3\rangle}(\usC_1(y_1,\tilde y),\varpibold_{2\alpha\alpha}y^{\alpha}y_2^{\alpha},\usC_3(y_3,\tilde y))\,,\\
    U_{\usC,\usC,\varpibold}&=+\int e^{-(1-v)\langle \p_1 \,\p_3\rangle +(1-u)\langle y\, \p_1\rangle -v\langle \p_2\,\p_3\rangle +u\langle y\,\p_2\rangle}(\usC_1(y_1,\tilde y),\usC_2(y_2,\tilde y),\varpibold_{3\alpha\alpha}y^{\alpha}y_3^{\alpha})\,.
    \end{align*}
    \end{subequations}
    To see the pattern of cancellation, we can bring $\varpibold_i$ in to the first position in $U$. Since $C$'s are 0-forms, the act of swapping positions of fields only provides minus signs in the exponential. Let the canonical ordering of fields in $\cU$ be $\varpibold,\usC,\usC$, we get
    \begin{align}\label{varpi-C-C}
        \cU(\varpibold,\usC,\usC)=&+\int e^{(1-u)\langle y\,\p_3\rangle + u\langle y\,\p_2\rangle }\big[2+(1-v)\langle \p_1\,\p_3\rangle + v\langle \p_1\,\p_2\rangle\big](\varpibold_1,\usC_2,\usC_3)\nn\\
        &-\int e^{(1-v)\langle y\,\p_1\rangle +v\langle y\,\p_3\rangle }\big[2+(1-u)\langle \p_2\,\p_1\rangle +u\langle \p_2\,\p_3\rangle \big](\varpibold_2,\usC_1,\usC_3)\nn\\
        &-\int e^{(1-u)\langle y\,\p_1\rangle +u\langle y\,\p_3\rangle }\big[2+(1-v)\langle \p_2\,\p_1\rangle +v\langle \p_2\,\p_3\rangle \big](\varpibold_2,\usC_1,\usC_3)\nn\\
        &+\int e^{(1-u)\langle y\,\p_1\rangle +u\langle y\,\p_2\rangle }\big[2+(1-v)\langle \p_3\,\p_1\rangle +v\langle \p_3\,\p_2\rangle \big](\varpibold_3,\usC_1,\usC_2)\,.\nn
    \end{align}
    By simple change in variables, say $u\leftrightarrow v$ in the first line of \eqref{varpi-C-C}, we see that all sub-vertices of $\cU$ add up to zero since they cancel each other pairwise. We conclude that $\cU(\varpibold,\usC,\usC)=0$, which is a beautiful consequence of chiral HSGRA being a Lorentz-invariant theory.

    $-$ \underline{$\ebold$ contributions.} In the case where $\cU(\Omegabold,\usC,\usC)\sim \cU(\ebold,\usC,\usC)$, 
    we obtain
    \begin{subequations}
        \begin{align*}
            U_{\ebold,\usC,\usC}&=+\Xi_3\big(2+[\p_1\,\p_2]+[\p_1\,\p_3]\big)\int e^{+(1-u)\langle y \,\p_3\rangle +u\langle y\p_2\rangle}(\ebold_{1\alpha\dot\alpha}y^{\alpha}\tilde y_1^{\dot\alpha},\usC_2,\usC_3)\,,\\
             U_{\usC,\ebold,\usC}&=-\Xi_3\big(2-[\p_2\,\p_1]+[\p_2\,\p_3]\big)\int e^{+(1-v)\langle y\, \p_1\rangle +v\langle y\p_3\rangle}(\ebold_{2\alpha\dot\alpha}y^{\alpha}\tilde y_2^{\dot\alpha},\usC_1,\usC_3)+ (v\leftrightarrow u)\,,\\
            U_{\usC,\usC,\ebold}&=+\Xi_3\big(2-[\p_3\,\p_1]-[\p_3\,\p_2]\big)\int e^{+(1-u)\langle y \,\p_1\rangle +u\langle y\,\p_2\rangle}(\ebold_{3\alpha\dot\alpha}y^{\alpha}\tilde y_3^{\dot\alpha},\usC_1,\usC_2)\,.
        \end{align*}
    \end{subequations}
\normalsize
where $ \Xi_3=\exp\big([y\,\p_1]+[y\,\p_2]+[y\,\p_3]\big)$. Observe that all $U_{\ebold,\usC,\usC}$ vertices cancel pairwise.

$\bullet$ For $n=4$ case, where
\begin{align*}   
    \cU(\Omegabold,\usC,\usC,\usC)=U_{\varpibold,\usC,\usC,\usC}+U_{\usC,\varpibold,\usC,\usC}+U_{\usC,\usC,\varpibold,\usC}+U_{\usC,\usC,\usC,\varpibold}\,,
\end{align*}
and \cite{Sharapov:2022nps} 
\small
\begin{subequations}\label{the-quartics}
\begin{align}
    U_{\varpibold,\usC,\usC,\usC}&=+\langle y\p_1\rangle^2\int e^{(1-u_1-u_2)\langle y\p_4\rangle + (1-v_1-v_2)\langle \p_1\p_4\rangle + u_1\langle y\p_2\rangle + u_2\langle y\p_3\rangle +v_1\langle \p_1\p_2\rangle + v_2\langle \p_1\p_3\rangle }\\
    U_{\usC,\varpibold,\usC,\usC}&=-\langle y\p_2\rangle^2\int e^{(1-u_1-u_2)\langle y\p_4\rangle +(1-v_1-v_2)\langle \p_2\p_4\rangle + u_2\langle y\p_1\rangle + u_1\langle y\p_3\rangle -v_2\langle \p_1\p_2\rangle +v_1\langle \p_2\p_3\rangle }\nn\\
    & \quad -\langle y\p_2\rangle^2\int e^{(1-u_1-u_2)\langle y\p_4\rangle +(1-v_1-v_2)\langle \p_2\p_4\rangle + u_1\langle y\p_1\rangle + u_2\langle y\p_3\rangle -v_1\langle \p_1\p_2\rangle +v_2\langle \p_2\p_3\rangle }\nn\\
    &\quad -\langle y\p_2\rangle^2\int e^{(1-u^R-v^L)\langle y\p_4\rangle +(1-u^L-v^R)\langle \p_2\p_4\rangle + v^L\langle y\p_1\rangle + u^R\langle y\p_3\rangle -u^L\langle \p_1\p_2\rangle +v^R\langle \p_2\p_3\rangle }\\
    U_{\usC,\usC,\varpibold,\usC}&=+\langle y\p_3\rangle^2\int e^{(1-u_1-u_2)\langle y\p_4\rangle + (1-v_1-v_2)\langle \p_3\p_4\rangle +u_2\langle y\p_1\rangle +u_1\langle y\p_2\rangle -v_2\langle \p_1\p_3\rangle -v_1\langle \p_2\p_3\rangle}\nn\\
    &\quad +\langle y\p_3\rangle^2\int e^{(1-v_1-v_2)\langle y\p_4\rangle + (1-u_1-u_2)\langle \p_3\p_4\rangle +v_1\langle y\p_1\rangle +v_2\langle y\p_2\rangle -u_1\langle \p_1\p_3\rangle -u_2\langle \p_2\p_3\rangle}\nn\\
    &\quad +\langle y\p_3\rangle^2\int e^{(1-u^R-v^L)\langle y\p_4\rangle + (1-u^L-v^R)\langle \p_3\p_4\rangle +v^L\langle y\p_1\rangle +u^R\langle y\p_2\rangle -u^L\langle \p_1\p_3\rangle -v^R\langle \p_2\p_3\rangle}\\
    U_{\usC,\usC,\usC,\varpibold}&=+\langle y\p_4\rangle^2\int e^{-(1-v_1-v_2)\langle\p_1\p_4\rangle + (1-u_1-u_2)\langle y\p_1\rangle -v_2\langle \p_2\p_4\rangle +u_2\langle y\p_2\rangle -v_1\langle \p_3\p_4\rangle +u_1\langle y\p_3\rangle }
\end{align}
\end{subequations}
\normalsize
with $(u_i,v_i,u^L,u^R,v^L,v^R)$ are certain coefficients that are not crucial in the followings.

Again, let the canonical ordering of the quartic in the spin-2 sector be $(\varpibold,\usC,\usC,\usC)$. We can ignore the pre-factors $\langle y\,\p_i\rangle^2$ of each $U$'s sub-vertices since they only produce $\varpibold_{\alpha\alpha}y^{\alpha(2)}$. It can then be shown in a simple manner that all $U$ sub-vertices cancel pairwise after some repositioning. 



Therefore, $\cV_{n \geq 4}(\Omegabold,\Omegabold,C,\ldots,C)=0$ and $\cU_{n\geq 3}(\Omegabold,C,\ldots,C)=0$ as claimed.  
\end{proof}



Let us now state our main results. 
\begin{theorem}\label{thm:eraseV(omega,omega,C)HS} Given $\uomega=\uomega([y\,\lambda],\langle y\,\lambda\rangle ,\langle \lambda\,\p'\rangle,\tH)$, then $\cV_{n\geq 3}(\uomega,\uomega,\usC,\ldots,\usC)=0$. 
\end{theorem}

\begin{proof} Recall that $\cV_n(\uomega,\uomega,\usC,\ldots,\usC)$ vertices have $\langle \p_i\,\p_j\rangle^{n-2}$ as prefactors where $i$th and $j$th are the positions of $\uomega_i,\uomega_j$ in the sub-vertices $V$'s. As $\uomega_{i}$ are functions of $\langle y\,\lambda_i\rangle$, it is clear that $\cV_n=0$ for $n\geq 3$. 
\end{proof}

\begin{theorem}\label{thm:eraseU(omega,C)HS} For $(\uomega,\usC)$ defined as above, all vertices $\cU_{n\geq 2}(\uomega,\usC,\ldots,\usC)=0$. 
\end{theorem}

\begin{proof} 
Consider $\cU_{n\geq 2}(\uomega,\usC,\ldots,C)$ whose sub-vertices $U$'s carry $\langle y\,\p_i\rangle^{n-2}$ as pre-factors. Then, all $U$'s will have the same pre-factors which are proportional to $(\langle y\,\lambda\rangle^{n-2}\times \ldots)$. 
At this stage, we can repeat the proof of Proposition \ref{prop:eraseV(Omega,C)} where we take the canonical ordering of fields in each $U$'s sub-vertices to be $(\uomega,\usC,\ldots,\usC)$. The claim is that all $U$'s vertices will cancel pairwise based on how vertices of chiral HSGRA were designed \cite{Sharapov:2022faa,Sharapov:2022awp,Sharapov:2022nps}. This statement can be checked explicitly with $\cU(\uomega,\usC)=U_{\uomega,\usC}+U_{\usC,\uomega}$, cf. \eqref{sub-U2}. 
It is analogous for the cubic \eqref{U3} (where we replace $\Omegabold$ by $\uomega$), and the quartic \eqref{the-quartics}. 
\end{proof}

With the above results, it is clear that the complete non-linear FDA is solved by $(\uomega,\usC)$. Thus, the non-linear FDA solved by $(\uomega,\usC)$ truncates to \eqref{eq:free-FDA} as it is not necessary to solve for higher-order fluctuations. 


\subsection{Effective action for chiral HSGRA on SD pp-wave background}\label{sec:action}
We now show that the effective action associated to the above exact solutions is precisely the usual kinetic action for chiral higher-spin fields on any SD background. 

 In what follows, we denote
\begin{align}
    \cF(\cA):=\Big(d\omega-\cV(\Omegabold,\omega)\Big)\Big|_{\tilde y=0}=D\cA\,,
\end{align}
to be the `effective curvature tensor' for chiral HSGRA, since we have shown all higher-order vertices vanish on the solutions. 
Then, the effective action associated to the above field strength reads
\begin{align}\label{eq:chiral-sigma-action}
    S_{\chi HSGRA}^{SD\, pp-wave}=\int_{\cM}\langle \sC\,|\,\cHbold\wedge D \cA\rangle_y\,,\qquad \cHbold:=\ebold_{\alpha\dot\alpha}\ebold_{\alpha}{}^{\dot\alpha}y^{\alpha}y^{\alpha}\,.
\end{align}
where $\langle\,\rangle_y$ stands for the operation of integrating out $y$ to produce appropriate contractions, cf. \cite{Krasnov:2021nsq}. We obtain
\small
\begin{align}
    S[B,A]=\frac{1}{2}\sum_{s=1}^{\infty}\int \left\langle \cB; \cHbold \wedge D\cA\right\rangle_y = \frac{1}{2}\sum_{s=1}^{\infty}\int \sqrt{g}B^{\alpha(2s)}D_{\alpha}{}^{\dot\alpha}A_{\alpha(2s-1)\,\dot\alpha}\,. 
\end{align}
\normalsize
Thus, the proposed action provides the kinetic terms for all chiral higher-spin fields on a generic SD background, as promised.




\section{Discussion}\label{sec:discussion}
In this paper, we obtain a class of exact self-dual pp-wave solutions, in chiral HSGRA by considering a SD modification of the Kerr-Schild ansatz, cf. \eqref{eq:chiral-ansatz-1}. 
We then generalize our analysis to the case of higher spins. For completeness, we show that the spacetime action of the obtained solutions reduces to the standard kinetic action for free fields on any SD background. 

Given that there are other self-dual solutions such as self-dual Taub-NUT black hole (see e.g. \cite{Hawking:1976jb,Boyer:1982mm,Adamo:2023fbj}), BPST instantons \cite{Belavin:1975fg}, or ADHM instantons \cite{Atiyah:1978ri}, we expect that there will be more exact solutions in chiral HSGRA, which are higher-spin generalizations of the aforementioned self-dual solutions.

The reason we believe there should be more classes of exact solutions in chiral HSGRA is that this theory is integrable in the sense of Bardeen \cite{Bardeen:1995gk}, see e.g. \cite{Ponomarev:2017nrr} and \cite{Monteiro:2022xwq} for relevant discussion. Therefore, chiral HSGRA should be solvable. This speculation is backed by the trivialization of the scattering amplitudes of chiral HSGRA in flat space, cf. \cite{Skvortsov:2018jea,Skvortsov:2020wtf,Skvortsov:2020gpn}, and some recent twistor construction of the theory or some closed subsectors of it, see e.g. \cite{Tran:2021ukl,Tran:2022tft,Herfray:2022prf,Neiman:2024vit}.

It will be interesting to check whether all self-dual/chiral higher-spin theories possess this property. Note that, even if some of these theories may have quantum anomalies, as discussed in \cite{Monteiro:2022xwq}, it is still possible to make them quantum integrable through a Green-Schwarz-like mechanism, as shown in \cite{Costello:2015xsa, Costello:2019jsy}. In particular, by introducing suitable axionic interactions in twistor space \cite{Costello:2021bah, Bittleston:2022nfr}, one should be able to cure any quantum anomalies that arise in self-dual/chiral theories. A beautiful discussion related to this story in spacetime can also be found, e.g., in \cite{Monteiro:2022nqt}. 

While the analysis of chiral HSGRA is relatively straightforward in $\Lambda=0$ case,\footnote{See also \cite{Basile:2024raj} for a higher-dimensional and \cite{Ponomarev:2024jyg} for a double-copy generalization of the theory.} it is more intricate in the case of (A)dS, where higher-order vertices may not vanish in a trivial way, as observed in \cite{Metsaev:2018xip,Skvortsov:2018uru,Sharapov:2022awp}. The reason for this is that the linear FDA \eqref{linear-FDA-general} contains additional terms associated with the cosmological constant, making it difficult to determine the appropriate ansatz for solving the linearized equation of motion. (In particular, when $\Lambda\neq 0$, one may need to introduce an additional pair of principal spinors to solve the linear FDA.) This fact makes the task of construction a covariant action for chiral HSGRA difficult. Nevertheless, one may hope that the light-cone ansatz introduced in \cite{Adamo:2022lah,Neiman:2023bkq} will provide further insight in guessing the correct action for chiral HSGRA due to the resulting simple form of higher-order vertices after substituting them in. See \cite{Chowdhury:2024dcy,Krasnov:2024qkh} for some work along this line.




\section*{Acknowledgement}
Many discussions with Yasha Neiman and Zhenya Skvortsov are gratefully acknowledged. The author also thanks an anonymous referee for useful suggestions. 
This work is supported by the Young Scientist Training (YST) program at the Asia Pacific Center of Theoretical Physics (APCTP) through the Science and Technology Promotion Fund and Lottery Fund of the Korean Government, and also the Korean Local Governments
– Gyeongsangbuk-do Province and Pohang City. 
The author also appreciates the support from the Fonds de la Recherche Scientifique under grant number F.4503.20 (HighSpinSymm), grant number F.4544.21 (HigherSpinGraWave) and funding from the European Research Council (ERC) under grant number 101002551, when part of this work took place in Mons, Belgium.

\appendix
\section{Convention and useful identities}
In the main text, we have raised and lowered spinorial indices with the following convention:
\begin{align}
    h^{\mu}_{\alpha\dot\alpha}h_{\mu\beta\dot\beta}=2\epsilon_{\alpha\beta}\epsilon_{\dot\alpha\dot\beta}\,,\quad v_{\alpha}=v^{\beta}\epsilon_{\beta\alpha}\,,\quad v^{\alpha}=v_{\beta}\epsilon^{\alpha\beta}\,,\quad v_{\dot\alpha}=v^{\dot\beta}\epsilon_{\dot\beta\dot\alpha}\,,\quad v^{\dot\alpha}=v_{\dot\beta}\epsilon^{\dot\alpha\dot\beta}\,.
\end{align}
Here, $\eps^{\alpha\beta}$ is the $\msp(2)$-invariant tensor  with the properties $\eps^{\alpha\beta}=\eps_{\alpha\beta}=-\eps^{\beta\alpha}$ and $\eps^{01}=1$. 
Furthermore, the following identities are useful
\begin{subequations}
\begin{align}
    h^{\alpha\dot\alpha}\wedge h^{\beta\dot\beta}&=+\frac{1}{2}H^{\alpha\beta}\eps^{\dot\alpha\dot\beta}+\frac{1}{2}\widetilde H^{\dot\alpha\dot\beta}\eps^{\alpha\beta}\,,\qquad \qquad &H^{\alpha\alpha}\wedge \widetilde H^{\dot\alpha\dot\alpha}&=0\,,\\
    H_{\alpha\alpha}\wedge H^{\alpha\alpha}&=-\widetilde H_{\dot\alpha\dot\alpha}\wedge \widetilde H^{\dot\alpha\dot\alpha}=\text{Vol}\,,\\
    h^{\alpha\dot\alpha}\wedge H^{\beta\beta}&=-\frac{1}{3}\eps^{\alpha\beta}\cH_3^{\beta\dot\alpha}\,,\qquad\qquad &h^{\alpha\dot\alpha}\wedge \widetilde H^{\dot\beta\dot\beta}&=+\frac{1}{3}\eps^{\dot\alpha\dot\beta}\cH_3^{\alpha\dot\beta}\,,
\end{align}
\end{subequations}
where
\begin{align}
     h^{\alpha}{}_{\dot\gamma}\wedge h^{\alpha\dot\gamma}=H^{\alpha\alpha}\,,\qquad \qquad  h_{\gamma}{}^{\dot\alpha}\wedge h^{\gamma\dot\alpha}=\widetilde H^{\dot\alpha\dot\alpha}\,,\qquad \qquad 
    h^{\alpha}{}_{\dot\alpha}\wedge H^{\dot\beta\dot\alpha}=\cH_3^{\alpha\dot\beta}\,.
\end{align}

\section{Strong homotopy algebras}\label{app:A}

This appendix provides some basic facts about homotopy algebras adapted to the notations used in the main text.

Recall that the $Q$-manifold $(\cN,Q)$ is a $\Z$-graded target space $\cN=\bigoplus_{n\in \Z}\cN_n$ whose homogeneous elements $a_n$ has degree $|a_n|=n$. Have we chosen to work with a degree $k$-shifted target space $\cN$, the corresponding subspaces would be
\begin{align}
    \cN[k]_n:=\cN_{n+k}\,.
\end{align}
Hereinafter, we use Koszul sign convention where $ab=(-)^{|a||b|}ba$ for two graded objects $(a,b)$ of degrees $|a|$ and $|b|$, respectively.

\paragraph{$A_{\infty}/L_{\infty}$-algebras.} Denote $C^{\bullet}(\cN)$ as the space of homomorphisms $\Hom_{\C}(T\cN,\cN)$, where $T\cN:=\bigoplus_{n=1}^{\infty}\cN^{\otimes n}$ is the tensor algebra associated with $\cN$, whose elements $\sF_m\in C^{m}(\cN)$ are $\sF_m$ maps with degree 1. We define the Gerstenhaber $\circ$-product as
\small
\begin{align}\label{eq:GH-product}
    \sF_n\circ \sF_m:=\sum_{k=0}^{n-1}(-)^{|\sF_m|(|\Phi_1|+\ldots+|\Phi_k|)}\sF_n(\Phi_1,\ldots,\Phi_k,\sF_m(\Phi_{k+1},\ldots,\Phi_{k+m}),\ldots,\Phi_{m+n-1})\,
\end{align}
\normalsize
where $\Phi$'s are elements of $\cN_0$ or $\cN_1$.

\begin{definition}\label{def:A-infinity} The pair $(\cN,\sF)$ where $\sF=\sum_{i= 1}^{\infty}\sF_i$ is called $A_{\infty}$-algebra if 
\begin{align}
    \sF\circ\sF=0\,.
\end{align}
\end{definition}
Following from the above definition of $A_{\infty}$-algebra, it is clear that $\sF_1^2=0$, and therefore $\sF_1$ defines a differential. The next structure map $\sF_2$ is a Leibniz-compatible binary operation. Then, the homotopy given by $\sF_3$ measures the failure of the associativity of $\sF_2$, and so on. For more detail, see e.g. \cite{Kajiura:2003ax}. 

Note that an $A_{\infty}$-algebra can induce an $L_{\infty}$-algebra defined by a pair ($\cN,\sL)$, where $\sL=\sum \sL_i$ with $|\sL|=1$ called $L_{\infty}$-structure maps with suitable symmetrization among the arguments, i.e.
\begin{align}
    \sL_n(\Phi_1,\ldots,\Phi_n)=\sum_{\sigma\in S_n}(-)^K\sF_n(\Phi_{\sigma(1)},\ldots,\Phi_{\sigma(n)})\,.
\end{align}
Here, $(-)^K$ is the Koszul sign resulting from $\sigma$ permutation under the symmetric group $S_n$. For instance,
\begin{align}
    \sL_n(\Phi_1,\ldots,\Phi_i,\Phi_{i+1},\ldots,\Phi_n)=(-)^{|\Phi_i||\Phi_{i+1}|}\sL_n(\Phi_1,\ldots,\Phi_{i+1},\Phi_i,\ldots,\Phi_n)\,.
\end{align}
As the reader may notice, $\sL$ can be identified with the homological vector field $Q$ on $\cN$. Furthermore, following from \eqref{eq:GH-product} and Definition \ref{def:A-infinity}, $\sL$'s must obey:
\begin{align}\label{generalized-Jacobi-L}
    \sum_{m+n-1=\text{const}}\sum_{\sigma\in S_n}(-)^K\sL_m(\sL_n(\Phi_{\sigma(1)},\ldots,\Phi_{\sigma(n)}),\Phi_{\sigma(n+1)},\ldots,\Phi_{\sigma(n+m-1)})=0\,.
\end{align}
As usual, $\sL_1\equiv d$ is the defining differential on $C^{\bullet}(\cN)$. Note that an $L_{\infty}$-algebra with $\sL_{n\geq3}=0$ is referred to as a dgLa. 

The $L_{\infty}$-structure maps $(\cV,\cU)$ for chiral HSGRA can be obtained by restricting the super coordinates $\Phi$'s to their appropriate components, i.e. to $\omega$'s or $\sC$'s. In particular, when there are strictly two $\omega$ arguments in $\sL$, then $\sL(\omega,\omega,\sC,\ldots,\sC)\equiv\cV(\omega,\omega,\sC,\ldots,\sC)$, and when there is only one $\omega$ we have the identification $\sL(\omega,\sC,\ldots,\sC)\equiv \cU(\omega,\sC,\ldots,\sC)$. 

Although $\cV$'s and $\cU$'s are $L_{\infty}$-structure maps, it will be more convenient to work with their sub-vertices which originate from a cyclic $A_{\infty}$-algebra where the signs and permutations can be kept in track in a simple manner, cf. \cite{Sharapov:2022faa}. This can be done by letting $\omega$ and $\sC$ take values in some matrix algebra $\Mat(n,\C)$ as alluded to in the main text. 
Furthermore, since $\cV$'s and $\cU$'s originated from the same $L_{\infty}$-algebra, we can relate them together through the following nature pairing.

\begin{definition}\label{def:cyclic-L-infinity} A degree-$2$ cyclic $L_{\infty}$-algebra is an $L_{\infty}$-algebra endowed with a non-degenerate $\C$-linear pairing $\langle -|-\rangle:\cN \otimes \cN^{\vee}[-1]\rightarrow \C[-2]$ such that
\begin{align}\label{K-pairing}
    \langle a|b\rangle=(-)^{(|a|+1)(|b|+1)}\langle b|a\rangle\,.
\end{align}
\end{definition}
From the above definition, we have for instance
\begin{subequations}\label{examples-of-traces}
\begin{align}
    \langle \sL_{n}(\Phi_1,\ldots,\Phi_n),\Phi_{n+1}\rangle&=(-)^{|\Phi_1|\big(\sum_{k=2}^{n+1}|\Phi_k|\big)+|\Phi_1|+|\Phi_{n+2}|}\langle \sL_n(\Phi_2,\ldots,\Phi_{n+2}),\Phi_1)\,,\\
    \langle \sL_n(\Phi_1,\ldots,\Phi_n)|\Phi_{n+1}\rangle&=(-)^{\sum_{k=1}^n|\Phi_k|}\langle \Phi_1|\sL_n(\Phi_2,\ldots,\Phi_{n+1})\rangle\,,
\end{align}
\end{subequations}
Upon restricting $\Phi$'s to $\omega$ and $\sC$, we obtain
\begin{subequations}\label{UV-duality-maps}
    \begin{align}
    \big\langle V(\omega,\sC,\ldots,\omega,\sC,\ldots,\sC)| C\big\rangle&=+\big\langle \omega|U(\sC,\ldots,\omega,\sC,\ldots,\sC)\big\rangle\,,\\
    \big\langle \sC| V(\sC,\ldots,\omega,\sC,\ldots,\sC,\omega)\big\rangle&=-\big\langle U(\sC,\ldots,\omega,\sC,\ldots,\sC)|\omega\big\rangle\,,
\end{align}
\end{subequations}
where $V$'s and $U$'s are sub-vertices originated from $A_{\infty}$-algebra of $\cV$ and $\cU$. The maps \eqref{UV-duality-maps} are called the $\cV$-$\cU$ duality maps \cite{Sharapov:2022faa,Sharapov:2022awp,Sharapov:2022nps}, which can be used to extract $\cU$'s vertices from $\cV$'s (at a fixed order $n$), and vice versa. 






\footnotesize
\bibliography{Blackhole.bib}
\bibliographystyle{JHEP-2}

\end{document}